\documentclass[a4paper, 11pt]{article}
\pdfoutput=1 
\usepackage{amsmath, amsfonts, amscd, amssymb,  amsthm, mathrsfs, ascmac}
\usepackage{latexsym}
\usepackage{longtable, geometry}
\usepackage[english]{babel}
\usepackage[utf8]{inputenc}
\usepackage{graphicx}
\usepackage{bbm}
\usepackage{enumitem}
\usepackage{nicefrac}
\usepackage{bm}
\usepackage[svgnames,psnames]{xcolor}
\usepackage[colorlinks,citecolor=DarkGreen,urlcolor=blue,linkcolor=FireBrick,linktocpage,unicode]{hyperref}
 \usepackage{braket}
\usepackage{xparse}
\usepackage{tikz-cd}
\usetikzlibrary{positioning}
\usepackage{authblk}
\usepackage{url, hyperref}

\newtheorem{Thm}{Theorem}[section]
\newtheorem{Def}[Thm]{Definition}
\newtheorem{Lemm}[Thm]{Lemma}
\newtheorem{Prop}[Thm]{Proposition}
\newtheorem{Coro}[Thm]{Corollary}
\theoremstyle{definition}

\newtheorem{Rem}[Thm]{Remark}

\newcommand{\be}{\begin{equation}}
\newcommand{\ee}{\end{equation}}

\newcommand{\ba}{\begin{align}}
\newcommand{\ea}{\end{align}}

\newcommand{\ben}{\begin{equation*}}
\newcommand{\een}{\end{equation*}}

\def\i<#1>{\langle #1 \rangle}
\def\l<#1>{\left\langle #1 \right\rangle}
\def\b<#1>{\big\langle #1 \big\rangle}
\def\wrt{\ \text{w.r.t.}\ }
\newcommand{\D}{\mathrm{dom}}
\newcommand{\la}{\langle}
\newcommand{\ra}{\rangle}
\newcommand{\bs}{\boldsymbol}

\newcommand{\Tr}{\mathrm{Tr}}
\newcommand{\BbbR}{\mathbb{R}}
\newcommand{\BbbN}{\mathbb{N}}
\newcommand{\BbbZ}{\mathbb{Z}}

\newcommand{\BbbC}{\mathbb{C}}
\newcommand{\bq}{{\boldsymbol q}}
\newcommand{\vepsilon}{\varepsilon}
\newcommand{\vphi}{\varphi}
\newcommand{\no}{\nonumber \\}

\newcommand{\up}{\uparrow}
\newcommand{\down}{\downarrow}
\newcommand{\h}{\mathfrak{H}}

\newcommand{\fP}{\mathfrak{P}}
\newcommand{\tH}{\tilde{H}}
\newcommand{\vLa}{\varLambda}

\def\i<#1>{\langle #1 \rangle}
\def\l<#1>{\left\langle #1 \right\rangle}
\def\b<#1>{\big\langle #1 \big\rangle}

\def\wrt{\ \text{w.r.t.}\ }

\makeatletter
  
  \@addtoreset{equation}{section}
\makeatother

\title{\sf
Rigorous analysis of the effects of electron-phonon interactions on magnetic properties in the one-electron Kondo lattice model

}

\date{}

\author[1]{Tadahiro Miyao}
\author[1]{Kazuhiro Nishimata}
\author[1]{Hayato Tominaga}

\affil[1]{Department of Mathematics,  Hokkaido University

Sapporo 060-0810,  Japan}

\begin{document}

\maketitle

\begin{abstract}
The Kondo lattice model (KLM) is a typical model describing heavy fermion systems.
In this paper, we consider the interaction of phonons with the system described by the one-electron KLM.
Magnetic properties of the ground state of this model are revealed in a rigorous form.
Furthermore, we derive the effective Hamiltonian in the strong coupling limit ($J\to \infty$) for the strength of the spin-exchange interaction $J$; we  examine the magnetic properties of the ground state of the effective Hamiltonian and prove that the Aizenman--Lieb theorem concerning the magnetization holds for the effective Hamiltonian at finite temperatures. Generalizing the obtained results, we clarify a mechanism for the stability of magnetic properties of the ground state in the one-electron KLM system.
\end{abstract}
\section{Introduction and results}
\subsection{Background}

The Kondo lattice model (KLM) is a typical model describing heavy fermion systems and has been actively investigated; see, e.g., \cite{Akagi2010,DONIACH1977231, PhysRevB.20.1969,Peters2007, Santos2002,Shen1996,TSU1997,Tsunetsugu1997}.
The KLM Hamiltonian consists of a hopping term for the conduction electrons and an exchange interaction term between the conduction electrons and the localized spins, and is derived from the strong coupling limit of the more basic Anderson model. 
In  \cite{PhysRevLett.67.2211}, the magnetic properties of the ground states of the KLM system have been analyzed in the particular case of a single conduction electron. This result is an essential clue to understanding the low electron concentration limit in the KLM system.
\medskip

In this paper, we consider the interaction of phonons with the system described by the one-electron KLM proposed in \cite{PhysRevLett.67.2211}. 
We show in a rigorous form that the magnetic properties of the one-electron KLM system described earlier are stable under the interaction with phonons. Furthermore, we derive the effective Hamiltonian in the strong coupling limit ($J\to \infty$) for the strength of the spin-exchange interaction $J$ and analyze the magnetic properties of the system described by the effective  Hamiltonian. To be more precise, we clarify the magnetic properties of the ground state of the effective Hamiltonian and prove that the Aizenman--Lieb theorem concerning the magnetization holds for this effective Hamiltonian at finite temperatures. 
Moreover, by synthesizing and generalizing these results, we reveal a  mechanism of stability of the magnetic properties of the one-electron KLM system.
Rigorous analysis of systems in which phonons interact with the KLM system has rarely been performed, and the results of this paper are expected to provide a solid mathematical foundation for this system.
\medskip

A characteristic feature of the analytical method in this paper is the application of the theory of  operator inequalities developed in \cite{Miyao2012,Miyao2016, Miyao2017,MIYAO2021168467}. 
In \cite{PhysRevLett.67.2211}, the magnetic properties of the ground state of the one-electron KLM are analyzed by naively using the Perron--Frobenius theorem for finite-dimensional Hilbert spaces; in contrast, the Hamiltonian examined in this paper acts on infinite-dimensional Hilbert spaces because it involves interaction with the phonons. Therefore, the ordinary Perron--Frobenius theorem cannot be applied to the Hamiltonian under consideration. On the other hand, our operator inequality-based analysis is adequate for rigorous analysis of ground states in infinite-dimensional Hilbert spaces.
\medskip

To highlight the novelty of this paper, we first review the previous work on the usual Kondo lattice model.
The Hamiltonian of the KLM on a  finite lattice $\vLa$ is
given by 
\be
H_{\rm KLM}=\sum_{x, y\in \vLa}(-t_{x, y}) c_{x, \sigma}^*c_{y, \sigma}+ J \sum_{x\in \vLa} {\bs s}_x\cdot {\bs S}_x-2 h S_{\rm tot}^{(3)}.
\ee
$H_{\rm KLM}$
is a bounded self-adjoint operator acting on the Hilbert space 
$
\bigoplus_{n=0}^{4|\vLa|} \bigwedge^n (\oplus^4\ell^2(\vLa))
$, where  $\bigwedge^n $ indicates the $n$-fold antisymmetric tensor product. 
The annihilation operators of conduction electrons are denoted by $c_{x, \sigma}$ and those of $f$-electrons by $f_{x, \sigma}$;
these operators satisfy the following anticommutation relations:
\begin{align}
	\{c_{x, \sigma}, c_{y, \tau}^*\} &= \delta_{x, y}\delta_{\sigma, \tau}, \quad
	\{c_{x, \sigma}, c_{y, \tau}\} = 0, \quad
	\{f_{x, \sigma}, f_{y, \tau}^*\} = \delta_{x, y}\delta_{\sigma, \tau}, \quad
	\{f_{x, \sigma}, f_{y, \tau}\} = 0, \\
	\{c_{x, \sigma}, f_{y, \tau}\} &= \{c_{x, \sigma}, f_{y, \tau}^*\} = 0
\end{align}
 for every $x, y \in \varLambda $ and $\sigma, \tau \in \{\up, \down\}$, where 
$\{X,Y\} := XY + YX$.
This paper considers a system with only one conduction electron and $|\vLa|$ $f$-electrons, one at each site.
The Hilbert space $\h_{\rm el}$ to adequately represent this situation is the subspace spanned by vectors of the form
\be
c_{x, \sigma}^*\prod^{\sharp}_{y\in \vLa} f_{y, \sigma_y}^*|\varnothing\ra, \quad
\  \sigma, \sigma_y\in \{\up, \down\}\ (y\in \vLa), 
\ee
where $|\varnothing\ra$ denotes the fermionic Fock vacuum, and $\prod^{\sharp}_{y\in \vLa}$ indicates the ordered product according to an arbitrarily fixed order in $\vLa$. 
Denote the restriction of the operator $H_{\rm KLM}$ to this Hilbert space by $H_{\rm el}$, i.e., $H_{\rm el}=H_{\rm KLM} \restriction \h_{\rm el}$.
For each site $x\in \vLa$, the spin operators ${\bs s}_x=(s_x^{(1)}, s_x^{(2)}, s_x^{(3)})$ of the conduction electrons are defined as follows: 
\begin{align}
s_x^{(1)}=\frac{1}{2}(c_{x, \down}^* c_{x, \up}+c_{x, \up}^*c_{x, \down}),\ \ s_x^{(2)}=\frac{\rm i}{2}(c_{x, \down}^* c_{x, \up}-c_{x, \up}^*c_{x, \down}), \ s_x^{(3)}=\frac{1}{2} (n_{x, \up}^c-n_{x, \down}^c), 
\end{align}
where $n_{x, \sigma}^c$ is the number operator for the conduction electrons with spin $\sigma$ at site $x\in \vLa$: $n_{x, \sigma}^c=c_{x, \sigma}^* c_{x, \sigma}$, and $\rm i$ represents the imaginary unit: ${\rm i}=\sqrt{-1}$.
Similarly, the spin operators ${\bs S}_x=(S_x^{(1)}, S_x^{(2)}, S_x^{(3)})$ for $f$-electrons at site $x$ are defined as
\begin{align}
S_x^{(1)}=\frac{1}{2}(f_{x, \down}^* f_{x, \up}+f_{x, \up}^*f_{x, \down}),\ \ S_x^{(2)}=\frac{\rm i}{2}(f_{x, \down}^* f_{x, \up}-f_{x, \up}^*f_{x, \down}), \ S_x^{(3)}=\frac{1}{2} (n_{x, \up}^f-n_{x, \down}^f),
\end{align}
where $n_{x,\sigma}^f=f_{x, \sigma}^* f_{x, \sigma}$. The total spin operators ${\bs S}_{\rm tot}=(S_{\rm tot}^{(1)}, S_{\rm tot}^{(2)}, S_{\rm tot}^{(3)})$ for this system are defined by
\be
S_{\rm tot}^{(i)}=\sum_{x\in \vLa} (s_x^{(i)}+S_x^{(i)}),\ \ i=1,2,3.
\ee
The spin exchange interaction term between the conduction electrons and the $f$-electrons is
defined  by 
\be
{\bs s}_x\cdot {\bs S}_x=\sum_{i=1,2,3} s_x^{(i)}S_x^{(i)}. \label{ExInt}
\ee
The term $2h S_{\rm tot}^{(3)}$ represents the interaction of the electrons with the external uniform magnetic field $h$.
$(t_{x, y})_{x, y\in \vLa}$ is the hopping amplitude, and $J$ is the strength of the exchange interaction.

In this paper, we make the following assumptions regarding the parameters and the lattice $\vLa$:
\begin{description}
\item[\hypertarget{A1}{(A. 1)}] 
\begin{itemize}
\item[(i)] $t_{x, y} \ge 0$ for every $x, y \in \vLa$.
\item[(ii)] $h\in \BbbR$.
\end{itemize}
	\item[(A. 2)]\hypertarget{A2}{} 
	Let $G=(\vLa, E)$ be the graph generated by the hopping matrix:
	$E=\{\{x, y\} : t_{x, y} \neq 0\}$ defines 
	the set of edges.  Then the graph $G$ is connected.\footnote{To be precise, for any $x, y\in \vLa $,  there is a sequence $\{\{x_i, x_{i+1}\}\}_{i=0}^{n-1}$  in  $E$ satisfying $x_0 = x $ and $ x_n = y$.}
	\end{description}
	
The following fundamental theorem is proved in \cite{PhysRevLett.67.2211}.	
	
\begin{Thm}\label{Senko1}
Assume \hyperlink{A1}{\bf (A. 1)} and \hyperlink{A2}{\bf (A. 2)}. Assume $J>0$ (antiferromagnetic coupling). Then,
the ground state of $H_{\rm el}$ is unique apart from the trivial $|\vLa|$-fold degeneracy and has total spin $S=(|\vLa|-1)/2$.
\end{Thm}

\subsection{The Kondo lattice system interacting with phonons}
It is an intriguing question how the magnetic properties of the ground state are affected when the interaction between the conduction electrons and the environmental system is taken into account.
This paper analyzes electron-phonon interactions in detail as a typical interaction with environmental systems.
We note that the analytical methods we use in this paper can be extended to systems in which electrons and Bose fields interact, such as systems in which electrons and quantized electromagnetic fields interact; see Section \ref{Discuss} for details.

The specific Hamiltonian we intend to analyze is given by 
\be
H=H_{\rm el}+ \sum_{x, y \in \vLa} g_{x, y}n_x^c(b_y + b_y^*) + \omega N_{\rm ph}.
\ee
The operator 
$H$ acts on the Hilbert space $\h=\h_{\rm el}\otimes \h_{\rm ph}$, where $\h_{\rm ph}$
is the bosonic Fock space over $\ell^2(\vLa)$: $\h_{\rm ph}=\bigoplus_{n=0}^{\infty} \otimes^n_{\rm s} \ell^2(\vLa)$, where
$\otimes^n_{\rm s}$ indicates the $n$-fold symmetric tensor product. 
$b_x^*$ and $b_x$ are the creation and annihilation operators of phonons, respectively, and satisfy the standard commutation relations:
\be
[b_x, b_y]=0,\quad [b_x, b_y^*]=\delta_{x, y},\quad x, y\in \vLa,\label{CCRs}
\ee
where $[X, Y]:=XY-YX$.
The operator $N_{\rm ph}$ denotes the phonon number operator:
\be
N_{\rm ph}=\sum_{x\in \vLa} b_x^*b_x.
\ee
The operator $n_x^c$ is the number operator of the conduction electrons at site $x\in \vLa$: $n_x^c=n_{x, \up}^c+n_{x, \down}^c$.
The phonons are assumed to be dispersionless with energy $\omega>0$.
$g_{x, y}\in \BbbR$ represents the strength of the interaction between conduction electrons and phonons.
Using Kato--Rellich's theorem \cite[Theorem X.12]{Reed1981}, we see that $H$ is a self-adjoint operator on $\D(N_{\rm ph})$, bounded from below, where
$\D(N_{\rm ph})$ denotes the domain of $N_{\rm ph}$.

Our first result is as follows:
\begin{Thm}\label{Main1}
Assume \hyperlink{A1}{\bf (A. 1)} and \hyperlink{A2}{\bf (A. 2)}. Assume $J>0$ (antiferromagnetic coupling). Then,
the ground state of $H$ is unique apart from the trivial $|\vLa|$-fold degeneracy and has total spin $S=(|\vLa|-1)/2$.
\end{Thm}
Comparing this theorem with Theorem \ref{Senko1}, we conclude that the magnetic properties of the ground state of $H_{\rm el}$ are stable under interaction with phonons. We will prove Theorem \ref{Main1} in Section \ref{Sec3}.

For each eigenvalue $M$ of $S^{(3)}_{\rm tot}$, we refer to $\h_M=\ker(S^{(3)}_{\rm tot}-M)$ as an $M$-subspace.
The Hilbert space $\h$ is decomposed into a direct sum of $M$-subspaces:
\be
\h=\bigoplus_{M\in \mathrm{spec}(S^{(3)}_{\rm tot})} \h_M,
\ee
where, for given operator $A$, $\mathrm{spec}(A)$ denotes the spectrum of $A$.
Corresponding to this decomposition, the Hamiltonian $H$ can be decomposed as follows: 
\be
H=\bigoplus_{M\in \mathrm{spec}(S^{(3)}_{\rm tot})} H_M,\quad H_M=H\restriction \h_M.
\ee
 In order to state the next result, we introduce the ladder operators of spin:
\begin{align}
	s_x^{(+)} &:= c_{x, \up}^*c_{x, \down}, &
	s_x^{(-)} &:= (s_x^+)^* = c_{x, \down}^*c_{x, \up}, \\
	S_x^{(+)} &:= f_{x, \up}^*f_{x, \down}, &
	S_x^{(-)} &:= (S_x^{(+)})^* = f_{x, \down}^*f_{x, \up}.
\end{align}

Our second result is the following theorem concerning two-point correlations between spins in the ground state:
\begin{Thm} \label{Main2}
Assume 
\hyperlink{A1}{\bf (A. 1)} and \hyperlink{A2}{\bf (A. 2)}.  Assume $J>0$ (antiferromagnetic coupling).
For any $M\in \{-(|\vLa|-1)/2, -(|\vLa|+1)/2, \dots, (|\vLa|-1)/2\}$,  let $\psi_M$ denote the normalized ground state of $H_M$ in the $M$-subspace, with $\la A\ra_M=\la \psi_M|A\psi_M\ra$ denoting the ground state expectation.
Then, for any $x, y\in\vLa$, we have
\begin{align}
&\Big\la s_x^{(+)} s_y^{(-)}\Big\ra_M>0, \quad \Big\la s_x^{(+)}S_y^{(-)}\Big\ra_M<0,\quad \Big\la S_x^{(+)}S_y^{(-)}\Big\ra_M>0,\\
&\Big\la s_x^{(-)} s_y^{(+)}\Big\ra_M>0, \quad \Big\la s_x^{(-)}S_y^{(+)}\Big\ra_M<0,\quad \Big\la S_x^{(-)}S_y^{(+)}\Big\ra_M>0.
\end{align}

\end{Thm}
The proof of Theorem \ref{Main2} will  be presented in Section \ref{Sec3}.

\begin{Rem}
For $J<0$ (ferromagnetic coupling),  Theorems \ref{Main1} and \ref{Main2} are modified as follows:
\begin{description}
\item[Theorem \ref{Main1}'] 
{\it 
Assume \hyperlink{A1}{\bf (A. 1)} and \hyperlink{A2}{\bf (A. 2)}. Assume $J<0$ (ferromagnetic coupling). Then,
the ground state of $H$ is unique apart from the trivial $|\vLa|+2$-fold degeneracy and has total spin $S=(|\vLa|+1)/2$.}
\item[Theorem \ref{Main2}'] {\it 
Assume \hyperlink{A1}{\bf (A. 1)} and \hyperlink{A2}{\bf (A. 2)}. Assume $J<0$ (ferromagnetic coupling).
For every $M \in \{-(|\vLa|-1)/2, -(|\vLa|+1)/2, \dots, (|\vLa|-1)/2\}$ and  $x, y\in\vLa$, we have}
\begin{align}
&\Big\la s_x^{(+)} s_y^{(-)}\Big\ra_M>0, \quad \Big\la s_x^{(+)}S_y^{(-)}\Big\ra_M>0,\quad \Big\la S_x^{(+)}S_y^{(-)}\Big\ra_M>0,\\
&\Big\la s_x^{(-)} s_y^{(+)}\Big\ra_M>0, \quad \Big\la s_x^{(-)}S_y^{(+)}\Big\ra_M>0,\quad \Big\la S_x^{(-)}S_y^{(+)}\Big\ra_M>0.
\end{align}
\end{description}
The proofs of these theorems are almost the same as the proofs of Theorems \ref{Main1} and \ref{Main2} and are therefore omitted.
\end{Rem}

\subsection{Strong coupling limit}
Consider the strong coupling limit of $J\to \infty$. 
In \cite{LACROIX1985991}, it is discussed that the usual KLM system without the electron-phonon interaction is equivalent to the Nagaoka--Thouless system in the strong coupling limit.
However, the arguments in \cite{LACROIX1985991}, while intuitively plausible, have the following mathematical difficulties: 
\begin{itemize}
\item since the ground state energy of $H$ diverges to $-\infty$ in the limit of $J\to \infty$, in order to rigorously realize the arguments of \cite{LACROIX1985991}, we need an energy renormalization procedure, as we will see below; 
\item furthermore, the Hamiltonian we consider contains unbounded operators describing the phonon, making mathematical analysis complicated.
\end{itemize}
 In this subsection,  we clarify how the arguments of \cite{LACROIX1985991} can be expressed mathematically for the Hamiltonian with the electron-phonon interaction, in the form of a theorem. 
 In the proof  in Section \ref{Sec4}, it will be shown how the issues mentioned above are overcome.
 We note that the results for the usual KLM can be derived by setting $g_{x, y}\equiv 0\, (x, y\in \vLa)$ in the theorems given below.

Set $\mathcal{S}_{\vLa}:=\{-1, +1\}^{|\vLa|}$.
For every $x\in \vLa$ and $\bs \sigma=(\sigma_x)_{x\in \vLa} $, define
\be
| {\bs \sigma}_x\ra=f_{x, \sigma_x} \prod_{x\in \vLa}^{\sharp} f_{x, \sigma_x}^*|\varnothing\ra, 
\ee
where the up-spin corresponds to $1$ and the down-spin corresponds to $-1$:
$\up=+1,\ \down=-1$.
This vector represents a state in which there is a hole at site $x$ and the spin configuration of the electrons at the other sites is given by ${\bs \sigma}_x:=(\sigma_y)_{y\in \vLa\setminus \{x\}}$.
Given 
${\bs \sigma}=(\sigma_y)_{y\in \vLa}\in \mathcal{S}_{\vLa}$, define
\be
|{\bs \sigma}_x\ra_{0}=
\frac{1}{\sqrt{2}} (c_{x\up}^* f_{x\down}^*-c_{x\down}^*f_{x\up}^*)
 |{\bs \sigma}_x\ra. \label{SingV}
\ee
This vector describes the situation where a conduction electron and an $f$-electron form a singlet at site $x$.
We readily confirm  that  
\be
{}_{0}\la {\bs \sigma}_x| {\bs \tau}_y\ra_{0}=\delta_{x, y} \delta_{{\bs \sigma}_x, {\bs \tau}_y}
\ee
holds.  Therefore, we see that $\{|{\bs \sigma}_x\ra_{0} : x\in \vLa,\ {\bs \sigma}\in \mathcal{S}_{\vLa}\}$ forms an orthonormal basis.
Put 
\be
\mathfrak{X}=\overline{\mathrm{span}}\{
| {\bs \sigma}_x\ra_0\otimes \vphi : {\bs \sigma}\in \mathcal{S}_{\vLa}, x\in \vLa, \vphi\in \h_{\rm ph}\}
\label{DefX}
\ee
and decompose the Hilbert space $\h$ as  $\h=\mathfrak{X}\oplus \mathfrak{X}^{\perp}$, where, for a  given  set $S$, $\mathrm{span}(S)$ indicates the linear span of $S$ and 
$\overline{\mathrm{span}}(S)$ denots the closure of $\mathrm{span}(S)$.
Henceforth, $P$ stands for the orthogonal projection from $\h$ to  $\mathfrak{X}$.

\begin{Thm}\label{Main3}
Define the energy renormalized Hamiltonian as 
\be
H_{{\rm ren}, J}=H+\frac{J}{2}. \label{DefHren}
\ee
There exists a self-adjoint operator $H_{\rm ren, \infty}$ on $\mathfrak{X}$, bounded from below, such that, for every $z\in \BbbC\setminus \mathrm{spec}(H_{\rm ren, \infty})$, it holds that 
\be
\lim_{J\to \infty}\Big\|(H_{{\rm ren}, J}-z)^{-1}-  (H_{\rm ren, \infty}-z)^{-1} P \Big\|=0,
\ee
\end{Thm}
\begin{Rem}\label{Neta}
\rm 
In Section \ref{Sec5},   $H_{{\rm ren}, \infty}$ is shown to be equivalent to the Nagaoka--Thouless Hamiltonian with added electron-phonon interaction.
\end{Rem}

The proof of Theorem \ref{Main3} is given in Section \ref{Sec4}.

For the system described by the  Hamiltonian $H_{{\rm ren}, \infty}$, some information on magnetic properties can be obtained.
In order to state our finding, we need the following condition:
\begin{description}
\item[\hypertarget{A3}{(A. 3)}] 
Let $G$ be the graph generated by the hopping matrix (see \hyperlink{A2}{\bf (A. 2)} for the detailed definition).
Then, $G$ is biconnected and it is not a simple loop (i.e., periodic chain) with more than four sites.\footnote{In other words, one cannot make $G$ disconnected by removing a single site.}
	\end{description}
	Regular lattices of two or more dimensions, such as triangular lattices, fcc, bcc, and $d$-dimensional hypercubic lattices ($d\ge 2$), satisfy this condition.

\begin{Thm}\label{Main4} 
 Assume \hyperlink{A1}{\bf (A. 1)} and \hyperlink{A3}{\bf (A. 3)}.
  Then,
the ground state of $H_{\rm ren, \infty}$ is unique apart from the trivial $|\vLa|$-fold degeneracy and has total spin $S=(|\vLa|-1)/2$.
\end{Thm}

\begin{Rem}
\rm 
The Hamiltonian $H_{\rm ren, \infty}$ treated here describes an extreme situation, and the reader may wonder about the significance of this theorem. In Section \ref{Discuss}, we clarify the importance of this theorem in the discussion of the stability of magnetic properties in the ground state of $H_{\rm ren,\infty}$.
\end{Rem}

We will prove Theorem \ref{Main4} in Section \ref{Sec5}.

The magnetization concerning $H_{\rm ren, \infty}$ is defined by
\be
M_{\rm ren, \infty}(\beta, h)=\frac{\partial \log Z_{\rm ren, \infty}}{\partial (\beta h)}, \label{DefMg}
\ee
where $Z_{\rm ren, \infty}$ is the partition function:
\be
Z_{\rm ren, \infty}=\Tr\big[e^{-\beta H_{\rm ren, \infty}}\big].
\ee

\begin{Thm}\label{Main5}
Let $d$ be a natural number greater than or equal to $2$ and consider the case $\vLa=[-L, L)^d\cap \BbbZ^d$. Suppose $(t_{x, y})$  represents   the nearest neighbor hopping matrix.\footnote{To be precise, $t_{x, y}=t>0$ if $\|x-y\|=1$, $t_{x, y}=0$, otherwise, where
$\|a\|=\sqrt{\sum_{i=1}^d|a_i|^2}$.}
Then we obtain
\be
M_{\rm ren, \infty}(\beta, h)
\ge (|\vLa|-1) \tanh (\beta h).
\ee
\end{Thm}
\begin{Rem}
\rm 
\begin{itemize}
\item
 This result is very similar to Aizemann--Lieb's theorem \cite{Aizenman1990}. This agreement follows from the facts stated in Remark \ref{Neta}. See Section \ref{Sec5} for details.
 \item Let $M(\beta, h)$ be the magnetization  defined  by replacing $H_{\rm ren, \infty}$ by $H$ in the definition of $M_{\rm ren, \infty}(\beta, h)$. From Theorem \ref{Main3}, it is expected that  $\lim_{J\to \infty}M(\beta, h)=M_{\rm ren, \infty}(\beta, h)$ holds.
   Indeed, in the case $g_{x, y} \equiv 0$,  this 
 immediately follows from Theorem \ref{Main3}.  In the case of $g_{x, y} \not\equiv 0$, however, showing  this seemingly straightforward relationship requires a rather complicated discussion of traces in infinite-dimensional Hilbert spaces, so we will not go into it further in this paper.
 \end{itemize}
\end{Rem}

The proof of Theorem \ref{Main5} will be  provided in Section \ref{Sec5}.

\subsection{Discussion: Stability of magnetic properties of the ground state} \label{Discuss}
The reader may have noticed that the results of Theorem \ref{Main1} and Theorem \ref{Main4} are quite similar.
Is this a coincidence?
Let us now examine the reasons for this coincidence from a higher perspective.
In a nutshell, the reason for this can be stated as follows:

\begin{Thm}\label{NTstatThm}
$H_{\rm el}, H$ and $ H_{\rm ren, \infty}$ all belong to the Nagaoka--Thouless stability class.
\end{Thm}

In the following, we will briefly explain the meaning of this theorem.
The Nagaoka--Thouless (NT) stability class is a set of Hamiltonians determined from the Nagaoka--Thouless Hamiltonian.
A notable feature of the NT stability class is that it has the following properties:

\begin{Thm}\label{NTGRP}
If a Hamiltonian $H_0$ belongs to the NT stability class, its ground state is unique apart from the trivial $|\vLa|$-fold degeneracy and has total spin $S=(|\vLa|-1)/2$.
\end{Thm}
From this fact and Theorem \ref{NTstatThm}, Theorems \ref{Senko1}, \ref{Main1} and \ref{Main4} follow immediately. For the precise definition of the NT stability class, see Appendix \ref{AppA}. Readers interested in the proof of Theorem \ref{NTstatThm}, see \cite{Miyao2019,Miyao2022}.\footnote{More precisely, one can prove Theorem \ref{NTstatThm} in almost the same way as the idea of \cite{Miyao2019,Miyao2022}.
}

The NT stability class describes the stability of the magnetic properties of the ground state. For example, consider the following Hamiltonian,  which further takes into account the interaction between $f$-electrons and phonons:
\begin{align}
H_{\rm fp}=H+ k \sum_{x\in \vLa}n_x^f
(b_x+b_x^*)+\nu\sum_{x\in \vLa}b_x^*b_x.
\end{align}
By the method of this paper, it can be shown that $H_{\rm fp}$ also belongs to the NT stability class. Using this and Theorem \ref{NTGRP}, we find that the ground state of $H_{\rm fp}$ is unique and has total spin $S=(|\vLa|-1)/2$. In general, if we can show that a Hamiltonian incorporating complicated interactions between electrons and the environment belongs to the NT stability class, we obtain a similar claim about the ground state. In addition to the examples mentioned above, the KLM, which incorporates the interaction between conduction electrons and quantized electromagnetic fields, also belongs to the NT-stability class. (For more details on this model, see, for example, \cite{Miyao2019,Miyao2022}.)

Next, a word of explanation on the importance of Theorem \ref{Main3} is in order.
Theorem \ref{Main3} and Remark \ref{Neta} show that $H_{\rm el}$ and $H$ are \lq\lq{}connected" with the NT Hamiltonian in the limit of $J\to \infty$.
This theorem provides a clue that $H_{\rm el}$ and $H$ belong to the NT stability class.
From these observations, we can conclude the following: although $H_{\rm ren, \infty}$ is a Hamiltonian describing a very extreme situation, its ground state properties are useful in analyzing the stability of the magnetic structure of the ground state of the KLM.

\subsection{Organization}
The paper is organized as follows: 
Section  \ref{Sec2} presents a brief overview of the theory of operator inequalities, which is necessary to prove the main theorems.
In Section \ref{Sec3}, we prove Theorems \ref{Main1} and \ref{Main2} by applying  the theory of operator inequalities described in Section \ref{Sec2}.
Section  \ref{Sec4} proves Theorem \ref{Main3}, and Section \ref{Sec5} proves Theorems \ref{Main4} and \ref{Main5}. 
Appendix \ref{AppA} is  supplementary to the discussions in Section \ref{Discuss}; the stability theory of magnetic properties in many-electron systems is outlined.

\subsection*{Acknowledgements}

T.M. was supported by JSPS KAKENHI Grant Numbers 18K03315, 20KK0304.

\section{Preliminaries}\label{Sec2}

This section defines the operator inequalities necessary to prove our main theorems and lists their basic properties. For more details on these operator inequalities, see \cite{Miyao2012,Miyao2016, Miyao2017,MIYAO2021168467}.
Note that the operator inequalities introduced here are different from the standard operator inequalities treated in functional analysis textbooks.

Let $\mathfrak{X}$ be a  complex separable Hilbert space.
We denote by $\mathscr{B}(\mathfrak{X})$ the Banach space of all bounded operators on $\mathfrak{X}$.

\begin{Def} \label{DefHilC}\upshape
A {\it Hilbert cone},  $\mathfrak{P}$  in $\mathfrak{X}$,   is  a closed convex  cone  obeying 
\begin{itemize}
\item[(i)] $\i<u|v>\geq 0$ for every  $u, v\in \mathfrak{P}$;
\item[(ii)] for each  $w\in \mathfrak{X}$, there exist  $u,u',v,v'\in \mathfrak{P}$ such that   $w=u-v+i(u'-v')$  and $\i<u|v>=\i<u'|v'>=0.$
\end{itemize}
A vector $ u \in\mathfrak{P}$ is said to be {\it  positive w.r.t.} $\mathfrak{P}$. We write this as $u \geq 0$ w.r.t. $\mathfrak{P}$. A vector $v \in\mathfrak{X}$ is called {\it strictly positive w.r.t.} $\mathfrak{P
}$,  whenever $\i<v|u>>0$ for all $ u \in \mathfrak{P}\setminus\{0\}$. We write this as $v>0$ w.r.t. $\mathfrak{P}$.
\end{Def}

The operator inequalities introduced below form the basis of the analytical methods in this paper.
\begin{Def}\upshape 
Let $A\in\mathscr B(\mathfrak{X})$. 
\begin{itemize}
\item[(i)] $A$ is {\it  positivity preserving} w.r.t.  $\mathfrak{P}$ if $A\mathfrak{P}\subseteq \mathfrak{P}$. We write this as $A\unrhd 0\wrt \mathfrak{P}$.
\item[(ii)] $A$ is {\it positivity improving} w.r.t. $\mathfrak{P}$ if,  for all $u \in \mathfrak{P
} \setminus \{0\},\ A u >0\wrt \mathfrak{P}$ holds.  We write this as $A\rhd 0\wrt \mathfrak{P}$.
\end{itemize}
Remark that the notations of the operator inequalities are  borrowed  from \cite{Miura2003}.
\end{Def}

The following lemma follows immediately from the definition:

\begin{Lemm}
Let $A,B\in\mathscr B(\mathfrak{X})$.
 Suppose that  $A\unrhd 0 $ and $B\unrhd0\wrt \mathfrak{P}$. We have the following {\rm (i)}--{\rm (iv)}:
 \begin{itemize}
 \item[\rm (i)] For each $u, v\in \mathfrak{P}$,  $\la u|Av\ra\ge 0$ holds.
 \item[{\rm (ii)}] If $a\ge 0 $ and $ b\ge 0$, then $aA+bB\unrhd0\wrt \mathfrak{P}$.
 \item[\rm (iii)] $A^*\unrhd 0$ w.r.t. $\mathfrak{P}$.
 \item[{\rm  (iv)}] $AB\unrhd0\wrt \mathfrak{P}$.
 \end{itemize}
\end{Lemm}
\begin{proof}
See, e.g., \cite{Miura2003,Miyao2016}. 
\end{proof}

\begin{Lemm}\label{Wcl}
Let $\{A_n\}_{n=1}^{\infty }$ and $A$ be 
bounded operators on $\mathfrak{X}$. If $A_n\unrhd 0$ w.r.t. $\fP$ and $A_n$ weakly converges to $A$ as $n\to \infty$, then $A\unrhd 0$ w.r.t. $\fP$ holds.
\end{Lemm}
\begin{proof}
See, e.g., \cite[Proposition A.1]{Miyao2020}.
\end{proof}

Let $\mathfrak{X}_{\mathbb{R}}$ be the real subspace of $\mathfrak{X}$ generated by $\mathfrak{P}$. From Definition \ref{DefHilC}, for all $x\in \mathfrak{X}_{\mathbb{R}}$, there exist $x_+, x_-\in \mathfrak{P}$ such that $x=x_+-x_-$ and $\langle x_+|x_-\rangle=0$. If $A\in \mathscr{B}(\mathfrak{X})$ satisfies $A\mathfrak{X}_{\mathbb{R}} \subseteq \mathfrak{X}_{\mathbb{R}}$, then we say that $A$ {\it preserves the reality w.r.t. $\mathfrak{P}$}.

\begin{Def} \upshape
Let $A, B\in\mathscr B(\mathfrak{X})$ be reality preserving w.r.t. $\mathfrak{P}$.  
If   $A-B\unrhd 0$ holds, then we write this as  
$A\unrhd B \wrt \mathfrak{P}$.
In this paper, we understand that $A$ and $B$ are always assumed to be reality preserving when one writes $A\unrhd B$ w.r.t. $\mathfrak{P}$.
\end{Def}

The following two lemmas are useful for practical applications of the operator inequalities introduced here.
\begin{Lemm}
Let $A,B,C,D\in\mathscr B(\mathfrak X)$. Suppose $A\unrhd B\unrhd0\wrt\mathfrak{P}$ and $C\unrhd D\unrhd0\wrt\mathfrak P$. Then we have $AC\unrhd BD\unrhd0\wrt\mathfrak P$.
\end{Lemm}
\begin{proof}
For proof, see, e.g., \cite{Miura2003,Miyao2016}. 
\end{proof}

\begin{Lemm}\label{ppiexp1}
Let $A,B$ be self-adjoint operators on $\mathfrak{X}$. Assume that $A$ is bounded from below and  that  $B\in\mathscr B(\mathfrak{X})$. Furthermore, suppose that $e^{-tA}\unrhd0\wrt\mathfrak{P}$ for all $t\geq0$ and $B\unrhd0\wrt\mathfrak{P}$. Then we have $e^{-t(A-B)}\unrhd e^{-tA}\wrt\mathfrak{P}$ for all $t\geq0$.
\end{Lemm}
\begin{proof}
Because  $B\unrhd0\wrt\mathfrak{P}$, we have $e^{tB}=\sum_{n=0}^{\infty}\frac{t^n}{n!}B^n\unrhd \mathbbm{1} \wrt\mathfrak{P}$ for all $t\geq0$. Note that we have used Lemma \ref{Wcl}.
By using  the Trotter product formula \cite[Theorem S. 20]{Reed1981},  we obtain
\begin{align}
e^{-t(A-B)}=\lim_{n\to\infty}\left(e^{-\frac{t}{n}A}e^{\frac{t}{n}B}\right)^n\unrhd e^{-tA}\wrt\mathfrak{P}
\end{align}
for all $t\geq0$, where  Lemma \ref{Wcl} is again used in deriving the  inequality.
\end{proof}

\begin{Def}
\upshape
Let $A$ be a self-adjoint operator on $\mathfrak{X}$,   bounded from below.
The semigroup  generated by $A$, $\{e^{-tA}\}_{t\ge 0}$,  is said to be {\it ergodic} w.r.t. $\mathfrak{P}$, if the following (i) and (ii) are satisfied:
\begin{itemize}
\item[(i)] $e^{-tA} \unrhd 0$ w.r.t. $\mathfrak{P}$ for all $t\ge 0$;
\item[(ii)] for each $u, v\in \mathfrak{P} \setminus \{0\}$, there is a $t\ge 0$
such that $\langle u| e^{-tA} v\rangle >0$. Note that $t$ could depend on $u$ and $v$.
\end{itemize}

\end{Def}

The following lemma   immediately follows  from the definitions:
\begin{Lemm}
Let $A$ be a self-adjoint operator on $\mathfrak{X}$, bounded from below. If $e^{-tA}\rhd 0$ w.r.t. $\mathfrak{P}$ for all $t>0$, then the semigroup  $\{e^{-tA}\}_{t\ge 0}$ is ergodic w.r.t. $\mathfrak{P}$.
\end{Lemm}

The following theorem illustrates why the operator inequalities introduced in this section are useful in this paper: 
\begin{Thm}[Perron--Frobenius--Faris]\label{pff}
Let $A$ be a self-adjoint operator, bounded from below. Assume that $E(A)=\inf \mathrm{spec}(A)$ is an eigenvalue of $A$.  If $\{e^{-tA}\}_{t\ge 0}$ is ergodic w.r.t. $\mathfrak{P}$, then $\dim \ker(A-E(A))=1$ and $\ker(A-E(A))$ is spanned by a   strictly positive vector w.r.t. $\mathfrak{P}$.
\end{Thm}

\begin{proof}
See \cite{Faris1972}.
\end{proof}

\section{Proof of Theorems \ref{Main1} and \ref{Main2}}\label{Sec3}

\subsection{Deformation of the Hamiltonian}
As a first step, one deforms the Hamiltonian $H$ into a form that is easy to analyze by a unitary transformation.

Let $N_{\down}^c $ be the number operator for the conduction electrons with down-spin: $N_{\down}^c=\sum_{x\in \vLa} n_{x, \down}^c$.
Performing the unitary transformation induced by $e^{{\rm i} \pi N_{\down}^c}$, $H_{\rm KLM}$ is transformed as follows:
\begin{align}
&e^{{\rm i} \pi N_{\down}^c} H_{\rm KLM} e^{-{\rm i} \pi N_{\down}^c}\no
=&\sum_{x, y\in \vLa} \sum_{\sigma=\up, \down}(-t_{x, y}) c_{x, \sigma}^*c_{y, \sigma}-J \sum_{x\in \vLa}  (s_x^{(+)}S_x^{(-)}+s_x^{(-)}S_x^{(+)})+J \sum_{x\in \vLa}s_x^{(3)} S_x^{(3)}-2h S^{(3)}_{\rm tot}.
\end{align}

Given $x\in \vLa$, define 
\begin{equation}
	p_x
	:= \frac{\rm i}{\sqrt{2}}(b_x^* - b_x), \qquad
	q_x
	:= \frac{1}{\sqrt{2}}(b_x^* + b_x).
\end{equation}
Both $p_x$ and $q_x$ are essentially self-adjoint. Therefore, we will also write the closures of these operators with the same symbols.
Next, define the anti-selfadjoint operator $L_c$ as 
\begin{equation}
	L_c
	:= -{\rm i}\frac{\sqrt{2}}{\omega}\sum_{x, y \in \vLa} g_{x, y}n_x^cp_y.
\end{equation}
The unitary transformation induced by the operator $e^{L_c}$ is called the {\it  Lang--Firsov transformation}.
The following formulas are helpful for specific calculations:
\begin{align}
		e^{L_c}c_{x, \sigma}e^{-L_c}
		= \exp\bigg({\rm i}\frac{\sqrt{2}}{\omega}\sum_{y \in \vLa} g_{x, y}p_y\bigg)c_{x, \sigma}, \quad
		e^{L_c}b_xe^{-L_c}
		= b_x - \frac{1}{\omega}\sum_{y \in \vLa} g_{y, x}n_y^c.
	\end{align}

With the above preparations, we transform $H$ as follows:
\begin{Lemm}
Define the unitary operator $F$ as $F=e^{L_c}e^{{\rm i}\pi N_{\down}^c}$. If we set $\tH=FHF^{-1}$, then the following holds:
\begin{align}
\tH
=& \sum_{x, y \in \vLa} \sum_{\sigma = \up, \down} (-t_{x, y})\exp({\rm i}\varPhi_{x, y})c_{x, \sigma}^*c_{y, \sigma} 
-J \sum_{x\in \vLa}  (s_x^{(+)} S_x^{(-)}+s_x^{(-)}S_x^{(+)})\no
			  &+J \sum_{x\in \vLa}s_x^{(3)} S_x^{(3)}-2h S^{(3)}_{\rm tot}+ \omega N_{\rm ph} - \sum_{x, y \in \vLa} G_{x, y}n_x^cn_y^c,  \label{CaltH}
\end{align}
where 
	\begin{equation}
		\varPhi_{x, y}
		= \frac{\sqrt{2}}{\omega}\sum_{z \in \vLa} (g_{y, z} - g_{x, z})p_z, \qquad
		G_{x, y}
		= \frac{1}{\omega}\sum_{z \in \vLa} g_{x, z}g_{y, z}.
	\end{equation}
\end{Lemm}
\begin{Rem}\rm 
Let us denote the restriction of $\tH$ to the $M$-subspace by $\tH_M$, that is, $\tH_M=\tH\restriction \h_M$.
Since $e^{{\rm i} \pi N^c_{\down}}$ commutes with $S_{\rm tot}^{(3)}$, 
$e^{{\rm i} \pi N^c_{\down}}$ remains $\h_M$ invariant.
However, since 
\be
FS_{\rm tot}^{(1)} F^{-1}=\sum_{x\in \vLa}(-s_x^{(1)}+S_x^{(1)}),\quad FS_{\rm tot}^{(2)} F^{-1}=\sum_{x\in \vLa}(-s_x^{(2)}+S_x^{(2)}), 
\ee
 we must be careful not to confuse the total spin of the ground state of $H_M$ with that of $\tH_M$.
\end{Rem}

\subsection{Strategy of the proof of Theorem \ref{Main1}}
In this section, we explain our strategy for proving Theorem \ref{Main1}.
First, note the following identification of the bosonic Fock space:
\be
\h_{\rm ph}=L^2(\mathcal{Q}),\qquad \mathcal{Q}=\BbbR^{|\vLa|}. \label{HilIdn}
\ee
Henceforth, this identification is frequently applied without any declaration.
Define the Hilbert cone in 
$\h_{\rm ph}$ by 
\be
L^2(\mathcal{Q})_+=\{\phi \in L^2(\mathcal{Q}) : \phi(\bq) \ge 0\ \mbox{a.e. $\bq$}\}.
\ee

For every $x \in \vLa, \sigma \in \{-1, 1\}$ and $ \bm{\sigma} = (\sigma_x)_{x\in \vLa} \in \mathcal{S}_{\vLa}$, we set 
\be
\ket{x, \sigma; \bm{\sigma}}
:= c_{x, \sigma}^*\prod_{y\in \vLa}^{\sharp}f_{y, \sigma_y}^* \ket{\varnothing}.
\ee
For each  $M\in {\rm spec}(S^{(3)}_{\rm tot})$, define
\be
\bm{\mathcal{S}}_{\vLa, M}=\bigg\{(\sigma, \bm{\sigma})\in \{-1, 1\}\times \mathcal{S}_{\vLa} : \sigma+\sum_{x\in \vLa} \sigma_x=2M\bigg\}.
\ee
With these preparations, if we put 
\be
\mathscr{C}_{\vLa, M}=\big\{ \ket{x, \sigma; \bm{\sigma}} : x \in \vLa, (\sigma, \bm{\sigma}) \in \bm{\mathcal{S}}_{\vLa, M}\big\}, 
\ee
 then $\mathscr{C}_{\vLa, M}$ forms a complete orthonormal system  (CONS) in the $M$-subspace $\h_{{\rm el}, M}$.
Now we define the Hilbert cone in 
$\h_{{\rm el}, M}$ by 
\be
\fP_{{\rm el}, M}=\mathrm{coni}(\mathscr{C}_{\vLa, M}), 
\ee
where, for  a given set $S$,  $\mathrm{coni}(S)$ stands for the conical hull of $S$.
Lastly, define the Hilbert cone in $\h_M$ as 
\begin{align}
\fP_M&=\overline{\mathrm{coni}}\big(\{\phi \otimes \psi : \phi\in \fP_{{\rm el}, M}, \psi\in L^2(\mathcal{Q})_+\}\big),
\end{align}
where $\overline{\mathrm{coni}}(S)$ denotes the closure of $\mathrm{coni}(S)$.

In Section \ref{PfPI}, we will prove the following theorem:
\begin{Thm}\label{PI}
For every $M\in  {\rm spec}(S^{(3)}_{\rm tot})$ and $\beta>0$, it holds that 
$e^{-\beta \tH_M} \rhd 0$ w.r.t. $\fP_M$. Especially, the semigroup $\{e^{-\beta \tH_M}\}_{\beta \ge 0}$ is ergodic w.r.t. $\fP_M$.
\end{Thm}

Once we accept this theorem, we can prove Theorem \ref{Main1}:

\subsubsection*{\it Proof of Theorem \ref{Main1}  given Theorem \ref{PI}}
By Theorems \ref{pff} and  \ref{PI}, the ground state of $\tH_M$ in each $M$-subspace is unique and can be chosen to be strictly positive w.r.t. $\fP_M$.
Thus, the ground state of $H_M$ in each $M$-subspace is unique.

With $M_{\rm max}=(|\vLa|+1)/2$, let us first examine the properties of the ground state $\psi_{M_{\rm max}}$ of $H_{M_{\rm max}}$.
Then, the ground state of $\tH_{M_{\rm max}}$ is given by $\tilde{\psi}_{M_{\rm max}} : =F\psi_{M_{\rm max}}$. Furthermore, $\tilde{\psi}_{M_{\rm max}}>0$ w.r.t. $\fP_{M_{\rm max}}$.

Define the vector in $\h_{M_{\rm max}}$   by 
\be
\chi=|\vLa|^{-1/2}\sum_{x\in \vLa}\ket{x, \up; \bm{\sigma}_{\up}}\otimes \phi,
\ee
where $\bm{\sigma}_{\up}=(\sigma_x)_x\in \mathcal{S}_{\vLa}$ is given by $\sigma_x=1$ for all $x\in \vLa$, and 
$\phi\in L^2(\mathcal{Q})_+$ is  a strictly positive  normalized vector.
We readily confirm  the following: (i) $\chi$ is strictly positive w.r.t. $\fP_{M_{\rm max}}$ and   normalized;  (ii) 
$F\chi=\chi$;
(iii)
$\chi$ has total spin $S=(|\vLa|+1)/2$. 

Imagine now that $ \psi_{M_{\rm max}}$ has total spin $S_0$. Because both $\chi$ and $\tilde{\psi}_{M_{\rm max}}$ are strictly positive w.r.t. $\fP_{M_{\rm max}}$, we find that the overlap between $\chi$ and $\psi_{M_{\rm max}}$ is strictly  positive, that is, 
\be
\la \chi| \psi_{M_{\rm max}}\ra=\la \chi| \tilde{\psi}_{M_{\rm max}}\ra>0,
\ee
which implies that 
\be
S(S+1) \la \chi| \psi_{M_{\rm max}}\ra=\la {\bs S}_{\rm tot}^2\chi| \psi_{M_{\rm max}}\ra
=\la \chi|  {\bs S}_{\rm tot}^2\psi_{M_{\rm max}}\ra=S_0(S_0+1)  \la \chi| \psi_{M_{\rm max}}\ra.
\ee
Hence, we conclude that $S_0=S=(|\vLa|+1)/2$.
Putting $M_{\min}=-(|\vLa|+1)/2$, we see in the same way that $\psi_{M_{\rm min}}$ also has total spin $(|\vLa|+1)/2$.

Next, let $M_{\dagger}=(|\vLa|-1)/2$. Define the vector $\eta$ in $\h_{M_{\dagger}}$ with 
\be
\eta=|\vLa|^{-1/2} \sum_{x\in \vLa}\ket{{\bs \sigma}_{ \up, x}}_0\otimes \phi,
\ee
where $\ket{{\bs \sigma}_{ \up, x}}_0$
 is the vector defined as $\bm{\sigma}=\bm{\sigma}_{\up}$ in \eqref{SingV} and $\phi$ is a strictly positive vector in $L^2(\mathcal{Q})$ as before. Note that $\eta$ has total spin $(|\vLa|-1)/2$. 
As $F\eta>0$ w.r.t. $\fP_{M_{\dagger}}$, it follows that $\la F\eta| \tilde{\psi}_{M_{\dagger}}\ra>0$.
Thus, by the \lq\lq{}overlap argument" used above, $\psi_{M_{\dagger}}$ has total spin $(|\vLa|-1)/2$.

Set 
$S^{(-)}_{\rm tot}=\sum_{x\in \vLa}(s_x^{(-)}+S_x^{(-)})$. If  we define $\eta_n=(S^{(-)}_{\rm tot})^n \eta\ (n=0, 1, \dots, |\vLa|-1)$, then  each 
$\eta_n$  belongs to $\h_{M_{\dagger}-n}$ and has total spin $(|\vLa|-1)/2$. Furthermore, since $F\eta_n$ is strictly positive w.r.t. $\fP_{M_{\dagger}-n}$, we  obtain
\be
\la \eta_n|\psi_{M_{\dagger}-n} \ra=\la F\eta_n|\tilde{\psi}_{M_{\dagger}-n}
\ra>0.
\ee
Thus, again by applying  the \lq\lq{}overlap argument", we conclude that $\psi_{M_{\dagger}-n}$ has total spin $(|\vLa|-1)/2$.
Let $E_M$ be the ground state energy of $H_M$.
From the conservation law for total spin, one finds the inequalities $E_{M_{\dagger}}<E_{M_{\rm max}}$ and $E_{M^{\dagger}}<E_{M_{\rm min}}$, and furthermore, when $|M| \le M_{\dagger}$, $E_M=E_{M_{\dagger}}$ follows.
We have thus completed the proof of Theorem \ref{Main1}.
\qed

\subsection{The heat semigroup generated by $\tH_M$ preserves the positivity}
In the following subsections, we prove Theorem \ref{PI}.
In order to achieve this goal, this subsection is devoted to proving the following proposition.

\begin{Prop}\label{PP}
For each $M\in  {\rm spec}(S^{(3)}_{\rm tot})$ and $\beta\ge 0$, 
 one obtains $e^{-\beta \tH_M} \unrhd 0$ w.r.t. $\fP_M$.
\end{Prop}

To show the proposition, we make some preparations.
The following lemma is often useful:

\begin{Lemm}\label{PBasic}
Let $A$ be a bounded linear operator on $\h_M$. Then, the following {\rm (i)} and {\rm (ii)} are equivalent to each other:
\begin{itemize}
\item[\rm (i)] $A\unrhd 0$ w.r.t. $\fP_M$.
\item[\rm (ii)] For each $x\in \vLa, (\sigma, \bm{\sigma})\in \bm{\mathcal{S}}_{\vLa, M}$ and $f\in L^2(\mathcal{Q})_+\setminus \{0\}$, it holds
$A\ket{ x, \sigma; \bm{\sigma}; f} \ge 0$ w.r.t. $\fP_M$, where
$\ket{ x, \sigma; \bm{\sigma}; f} =\ket{ x, \sigma; \bm{\sigma}}\otimes f$.
\end{itemize}
\end{Lemm}
\begin{proof}
(i) $\Rightarrow$ (ii): This part is trivial.

(ii) $\Rightarrow$ (i): Any $\vphi\in \fP_M$ can be expressed as
\be
\vphi=\sum_{n=1}^N \ket{x_n, \sigma_n; \bm{\sigma}_n; f_n},
\ee
where $(\sigma_n, \bm{\sigma}_n)\in \bm{\mathcal{S}}_{\vLa, M}$ and $f_n\in L^2(\mathcal{Q})_+\setminus\{0\}$.
Using   (ii), we have $A \ket{x_n, \sigma_n; \bm{\sigma}_n; f_n}\ge 0$ w.r.t. $\fP_M$,
which implies 
$
A\vphi=\sum_{n=1}^N A \ket{x_n, \sigma_n; \bm{\sigma}_n; f_n}\ge 0\ \mbox{w.r.t. $\fP_M$}.
$
\end{proof}

\begin{Lemm}\label{IroiroPP}
For all $M\in {\rm spec}(S_{\rm tot}^{(3)})$, we have the following:
\begin{itemize}
\item[\rm (i)] $c_{x, \sigma}^* c_{y, \sigma} \unrhd 0$ w.r.t. $\fP_M$ for each $x, y\in \vLa$ and $\sigma\in \{-1, 1\}$.
\item[\rm (ii)] $s_x^{(+)} S_x^{(-)} \unrhd 0$ and $s_x^{(-)} S_x^{(+)}\unrhd 0$ w.r.t. $\fP_M$ for each $x\in \vLa$.
\item[\rm (iii)] $e^{{\rm i} \varPhi_{x, y}} \unrhd 0$ w.r.t. $\fP_M$ for each $x, y\in \vLa$.
\item[\rm (iv)] $e^{-\beta N_{\rm ph}} \rhd 0$ w.r.t. $\fP_M$ for all $\beta>0$.
\end{itemize}
\end{Lemm}
\begin{proof}
For each $x\in\vLa$, let the map $T_x : \mathcal{S}_{\vLa}\to \mathcal{S}_{\vLa}$ be the spin-flip at site $x$: for each $\bm{\sigma} \in\mathcal{S}_{\vLa}$, $(T_x\bm{\sigma})_y=\sigma_y$ if $y\neq x$, $(T_x\bm{\sigma})_y=-\sigma_x$ if $x=y$.
First, we note the following:
\begin{align}
c_{x, \sigma}^* c_{y, \sigma}|z,  \mu;  \bm{\sigma}\ra&=
\begin{cases}
0 & \mbox{if $(y, \sigma)\neq (z, \mu)$}\\
\ket{x, \sigma; \bm{\sigma}}  & \mbox{if $(y, \sigma)= (z, \mu)$,}
\end{cases}\label{Cact}\\
c_{x, \up}^* c_{x, \down}|z,  \mu;  \bm{\sigma}\ra&=
\begin{cases}
0 & \mbox{if $(x, \down)\neq (z, \mu)$}\\
\ket{x, \up; \bm{\sigma}}  & \mbox{if $(x, \down)= (z, \mu)$,}
\end{cases}\\
c_{x, \down}^* c_{x, \up}|z,  \mu;  \bm{\sigma}\ra&=
\begin{cases}
0 & \mbox{if $(x, \up)\neq (z, \mu)$}\\
\ket{x, \down; \bm{\sigma}}  & \mbox{if $(x, \up)= (z, \mu)$}
\end{cases}
\end{align}
and 
\begin{align}
f_{x, \up}^*f_{x, \down} \ket{ y, \sigma; \bm{\sigma}}&=
\begin{cases}
\ket{y, \sigma; T_x \bm{\sigma}} & \mbox{if $\sigma_x=-1$}\\
0  & \mbox{if $\sigma_x=1$},
\end{cases} \\
f_{x, \down}^*f_{x, \up} \ket{ y, \sigma; \bm{\sigma}}&=
\begin{cases}
0 & \mbox{if $\sigma_x=-1$}\\
\ket{y, \sigma; T_x \bm{\sigma}}  & \mbox{if $\sigma_x=1$}.
\end{cases}
\end{align}
(i) Using \eqref{Cact}, 
 we see $c_{x, \sigma}^* c_{y, \sigma}|z,  \mu; \bm{\sigma}; f\ra\ge 0$ w.r.t. $\fP_M$
  for each $x\in \vLa, (\sigma, \bm{\sigma})\in \bm{\mathcal{S}}_{\vLa, M}$ and $f\in L^2(\mathcal{Q})_+\setminus \{0\}$.  Hence,  we conclude (i) by applying  Lemma \ref{PBasic}.
Similarly, we can prove (ii).

(iii)
Under the identification \eqref{HilIdn}, we can identify $q_x$ with the multiplication operator on $L^2(\mathcal{Q})$, and $p_x$ with the partial differential operator $-{\rm i} \partial /\partial q_x$.
Thus, $e^{{\rm i} a p_x}\ (a\in \BbbR, x\in \vLa)$ can be regarded as a translation operator:
$(e^{{\rm i} a p_x} f)(\bm{q})=f(\bm{q}+a\delta_x)\, (f\in L^2(\mathcal{Q}))$. 
If $f(\bm{q}) \ge 0$ a.e., then $(e^{{\rm i} a p_x} f)(\bm{q})\ge 0$ a.e. holds, which implies that $e^{{\rm i} a p_x} \unrhd 0$ w.r.t. $L^2(\mathcal{Q})_+$.
Therefore, we conclude  that $ e^{{\rm i} a p_x}|z,  \mu;  \bm{\sigma}; f\ra=|z,  \mu;  \bm{\sigma};  e^{{\rm i} a p_x}f\ra \ge 0$ w.r.t. $\fP_M$.

(iv) Again using the identification \eqref{HilIdn}, we can express $N_{\rm ph}$ as $N_{\rm ph}=\frac{1}{2}\sum_{x\in \vLa}(-\Delta_{q_x}+q_x^2-1)$, a Hamiltonian of harmonic oscillators. 
As is well-known, the kernel of the heat semigroup generated by  the  Hamiltonian of harmonic oscillator
 is strictly positive, see, e.g., \cite{Simon2005}.
Using this fact,  we immediately obtain (iv).
\end{proof}

Decompose the operator $\tH$ as
\be
\tH=-D+V+\omega N_{\rm ph}, \label{tHDec}
\ee
where
\begin{align}
D&=\sum_{x, y\in \vLa}\sum_{\sigma=-1, 1}t_{x, y}\exp({\rm i} \varPhi_{x, y}) c_{x, \sigma}^*c_{y, \sigma}+J \sum_{x\in \vLa}  (s_x^{(+)} S_x^{(-)}+s_x^{(-)}S_x^{(+)}),
\\
V&=J \sum_{x\in \vLa}s_x^{(3)} S_x^{(3)}-2h S^{(3)}_{\rm tot}- \sum_{x, y \in \vLa} G_{x, y}n_x^cn_y^c.
\end{align}
Note the minus sign in front of $D$ in \eqref{tHDec}.

\begin{Lemm}\label{IroiroPP2}
For every  $M\in \mathrm{spec(S_{\rm tot}^{(3)})}$, 
we have the following {\rm (i)} and {\rm (ii)}:
\begin{itemize}
\item[\rm (i)] $D\unrhd 0$ w.r.t. $\fP_M$.
\item[\rm (ii)] $e^{-\beta V} \unrhd 0$ w.r.t. $\fP_M$ for all $\beta \ge 0$.
\end{itemize}
\end{Lemm}
\begin{proof}
(i)
By applying (i) and (iii) of Lemma \ref{IroiroPP}, we obtain 
$\exp({\rm i}\varPhi_{x, y})
c_{x, \sigma}^* c_{y, \sigma} \unrhd 0
$ w.r.t. $\fP_M$.
Using (ii)  of Lemma \ref{IroiroPP}, we see that $s_x^{(+)} S_x^{(-)}+s_x^{(-)}S_x^{(+)} \unrhd 0$  w.r.t. $\fP_M$. Putting these together, we can conclude (i).

(ii) As can be easily seen, $|x, \sigma; \bm{\sigma}\ra$ is an eigenvector of $V$.
So, denoting the corresponding eigenvalue as $V(x, \sigma; \bm{\sigma})$,
we have
\be
e^{-\beta V} \ket{x, \sigma; \bm{\sigma}; f}=e^{-\beta V(x, \sigma; \bm{\sigma})}\ket{x, \sigma; \bm{\sigma}; f} \ge 0\ \mbox{w.r.t. $\fP_M$}
\ee
for all 
$f\in L^2(\mathcal{Q})_+\setminus \{0\}$.
Hence, by applying Lemma \ref{PBasic}, we conclude (ii).

\end{proof}

\subsubsection*{\it Proof of Proposition \ref{PP}}
Because $V$ commutes with $N_{\rm ph}$, we have,
 by using Lemmas \ref{IroiroPP} and \ref{IroiroPP2}, \be
 e^{-\beta (V+\omega N_{\rm ph})}=e^{-\beta V} e^{-\beta \omega N_{\rm ph}} \unrhd 0 \ \mbox{w.r.t. $\fP_M$ for all $\beta\ge 0$.}
 \ee
Hence, by applying Lemma \ref{ppiexp1} with $A=V+\omega N_{\rm ph}$ and $B=D$, we obtain the desired assertion. \qed

\subsection{The heat semigroup generated by $\tH_M$ improves  the positivity}\label{PfPI}
Let $G=(\vLa, E)$ be the graph generated by the hopping matrix $(t_{x, y})$;  see \hyperlink{A2}{\bf (A. 2)} for details.
For each $\{x, y\}\in E$ and $x\in \vLa$, set 
\be
D^{(0)}_{x, y} =\sum_{\sigma=-1,  1}t_{x, y}\exp({\rm i} \varPhi_{x, y})c_{x, \sigma}^*c_{y, \sigma}, \quad D_x^{(1)}=J (s_x^{(+)} S_x^{(-)}+s_x^{(-)}S_x^{(+)}).
\ee
For any  given path $\bm{p}=(\{x_i, x_{i+1}\})_{i=1}^{n-1}\subset E$, we define the operator
$D_{\bm{p}}^{(0)}$ by 
\be
D_{\bm{p}}^{(0)}=D^{(0)}_{x_1, x_2} D^{(0)}_{x_2, x_3} \cdots D_{x_{n-1}, x_n}^{(0)}.
\ee

The purpose of this subsection is to prove the following proposition:
\begin{Prop}\label{Reduction}
A sufficient condition for $e^{-\beta \tH_M} \rhd 0$ w.r.t. $\fP_M$ to  be valid for all $\beta>0$  is that the following condition holds:
\begin{description}
\item[\bf \hypertarget{R}{(R)}]
For every   $x, y\in \vLa$, $(\sigma, \bm{\sigma}), (\tau, \bm{\tau})\in \bm{\mathcal{S}}_{\vLa, M}$, $f\in L^2(\mathcal{Q})_+\setminus \{0\}$ and 
$g\in L^2(\mathcal{Q})_+\setminus \{0\}$, strictly positive, there exist paths $\bm{p}_1, \dots, \bm{p}_n$ and sequence $\{x_i\}_{i=1}^{n-1}\subset \vLa$  such that 
\be
\bra{x, \sigma; \bm{\sigma}; f} (D^{(0)}_{\bm{p}_1} D_{x_1}^{(1)})( D^{(0)}_{\bm{p}_2} D_{x_1}^{(1)}) \cdots (D^{(0)}_{\bm{p}_{n-1}} D_{x_{n-1}}^{(1)}) D^{(0)}_{\bm{p}_{n}} \ket{y, \tau; \bm{\tau}; g}>0. \label{Dprod}
\ee
\end{description}
\end{Prop}

The proof consists of four steps.
\subsubsection*{Step 1}
By applying Duhamel's expansion, we have
\be
e^{-\beta \tH}=\sum_{n=0}^{\infty} \mathscr{D_n(\beta)}, \label{Duha}
\ee
where the right hand side of \eqref{Duha} converges in the uniform topology  and  $\mathscr{D}_n(\beta)$ are defined by 
\be
\mathscr{D}_n(\beta)=\int_{\Delta_n(\beta)} D(s_1)\cdots D(s_n) e^{-\beta L} ds_1\cdots ds_n,
\ee
where $\Delta_n(\beta)=\{(s_1, \dots, s_n)\in \BbbR^n : 0\le s_1\le \cdots \le s_n \le \beta\}$,  $L=V+\omega N_{\rm ph}$ and $D(s)=e^{-sL} De^{sL}$.
By using  Lemma \ref{IroiroPP2}, we readily confirm 
\be
D(s_1)\cdots D(s_n) e^{-\beta L} \unrhd 0\quad \mbox{w.r.t. $\fP_M$},
\ee
provided that $(s_1, \dots, s_n)\in \Delta_n(\beta)$.
Combining this with  Lemma \ref{Wcl}, we obtain $\mathscr{D}_n(\beta) \unrhd 0$ w.r.t. $\fP_M$. Therefore, for every $n\in \BbbZ_+$, we have
\be
e^{-\beta H} \unrhd \mathscr{D}_n(\beta),
\ee
which implies the following lemma:

\begin{Lemm}\label{R1Lemma}
A sufficient condition for $e^{-\beta \tH_M} \rhd 0$ w.r.t. $\fP_M$ to  be valid for all $\beta>0$  is that the following condition holds:
\begin{description}
\item[\bf \hypertarget{R1}{(R. 1)}]
For every  $\vphi, \psi\in \fP_M\setminus \{0\}$ and $\beta>0$, there exists an $n\in \BbbZ_+$ such that 
$\la \vphi|\mathscr{D}_n(\beta)\psi\ra>0$. 
\end{description}
\end{Lemm}

\subsubsection*{Step 2}
For simplicity of presentation, set 
\be
\mathscr{D}_n(\beta; \bm{s})=D(s_1)\cdots D(s_n) e^{-\beta L}\quad (\bm{s}=(s_1, \dots, s_n)\in \Delta_n(\beta)).
\ee

\begin{Lemm}\label{R2Lemma}
A sufficient condition for {\bf \hyperlink{R1}{(R. 1)}} to  be valid   is that the following condition holds:
\begin{description}
\item[\bf \hypertarget{R2}{(R. 2)}]
For every  $\vphi, \psi\in \fP_M\setminus \{0\}$ and $\beta>0$, there exists an $n\in \BbbZ_+$ such that 
$\la \vphi|\mathscr{D}_n(\beta; \bm{0}) \psi\ra>0$, where $\bm{0}=(0, \dots, 0)\in \Delta_n(\beta)$. 
\end{description}
\end{Lemm}
\begin{proof}
First, remark the fact that  $\mathscr{D}_n(\beta; \bm{s}) \unrhd 0$ w.r.t. $\fP_M$ for every $\bm{s} \in \Delta_n(\beta)$ implies $\la \vphi|\mathscr{D}_n(\beta; \bm{s}) \psi\ra \ge 0$.
Because $\mathscr{D}_n(\beta; \bm{s})$ is continuous in $\bm{s}$, we get
\be
\la \vphi|\mathscr{D}_n(\beta)\psi\ra=\int_{\Delta_n(\beta)} \la \vphi|\mathscr{D}_n(\beta; \bm{s}) \psi\ra d\bm{s}>0.
\ee
Thus, we are done.
\end{proof}

\subsubsection*{Step 3}

To describes the next step, we divide $D$ as $D=D^{(0)}+D^{(1)}$, where
\be
D^{(0)} =\sum_{x, y\in \vLa}\sum_{\sigma=-1,  1}t_{x, y}\exp({\rm i} \varPhi_{x, y})c_{x, \sigma}^*c_{y, \sigma}, \quad D^{(1)}=J \sum_{x\in \vLa}  (s_x^{(+)}S_x^{(-)}+s_x^{(-)}S_x^{(+)}).
\ee

\begin{Lemm}\label{R3Lemma}
A sufficient condition for {\bf \hyperlink{R2}{(R. 2)}} to  be valid   is that the following condition holds:
\begin{description}
\item[\bf \hypertarget{R3}{(R. 3)}]
For every   $x, y\in \vLa$, $(\sigma, \bm{\sigma}), (\tau, \bm{\tau})\in \bm{\mathcal{S}}_{\vLa, M}$, $f\in L^2(\mathcal{Q})_+\setminus \{0\}$ and 
$g\in L^2(\mathcal{Q})_+\setminus \{0\}$, strictly positive, there exist $n\in \BbbZ_+$ and  $\vepsilon_1, \dots, \vepsilon_n\in \{0, 1\}$ such that 
\be
\bra{x, \sigma; \bm{\sigma}; f} D^{(\vepsilon_1)} \cdots D^{(\vepsilon_n)} \ket{y, \tau; \bm{\tau}; g}>0.
\ee
\end{description}
\end{Lemm}
\begin{proof}
For each $\eta\in \fP_M\setminus \{0\}$, there exist $x\in \vLa, (\sigma, \bm{\sigma})\in \bm{\mathcal{S}}_{\vLa, M}$ and $f\in L^2(\mathcal{Q})_+\setminus \{0\}$ such that 
\be
\eta \ge \ket{x, \sigma, \bm{\sigma}; f} \quad\mbox{w.r.t. $\fP_M$}.
\ee
To see this, we just recall that $\eta$ can be expressed as 
$ \eta=\sum_{i=1}^N\ket{x_i, \sigma_i; \bm{\sigma}_i; f_i}$ with $f_i\in L^2(\mathcal{Q})_+\setminus \{0\}$.

From the above discussion, we see that there exist $x, y\in \vLa, (\sigma, \bm{\sigma}), (\tau, \bm{\tau})\in \bm{\mathcal{S}}_{\vLa, M}$ and $f, g\in L^2(\mathcal{Q})_+\setminus \{0\}$ such that 
\be
\vphi\ge \ket{x, \sigma, \bm{\sigma}; f},  \quad \psi\ge \ket{y, \tau, \bm{\tau}; f} \quad\mbox{w.r.t. $\fP_M$},
\ee
which implies that 
\be
\la \vphi|\mathscr{D}_n(\beta; \bm{0}) \psi\ra\ge
\bra{x, \sigma; \bm{\sigma}; f} \mathscr{D}_n(\beta; \bm{0}) \ket{y, \tau; \bm{\tau}; g}.
\ee

On the other hand, since $e^{-\beta L}\ket{y, \tau; \bm{\tau}; g}=e^{-\beta V(y, \tau; \bm{\tau})}\ket{y, \tau; \bm{\tau}; e^{-\beta \omega N_{\rm ph}}g}$, we obtain 
\be
\bra{x, \sigma; \bm{\sigma}; f} \mathscr{D}_n(\beta; \bm{0}) \ket{y, \tau; \bm{\tau}; g}=e^{-\beta V(y, \tau; \bm{\tau}) }
\bra{x, \sigma; \bm{\sigma}; f} D^n \ket{y, \tau; \bm{\tau}; e^{-\beta \omega N_{\rm ph}}g},  \label{EquivDn}
\ee
where $V(y, \tau; \bm{\tau})$ is an eigenvalue of $V$; see the proof of Lemma \ref{IroiroPP2} for details. Because $D\unrhd D^{(\vepsilon)}$ w.r.t. $\fP_M\ (\vepsilon=0, 1)$,  we find 
\be
\mbox{the right hand side of \eqref{EquivDn} } \ge e^{-\beta V(y, \tau; \bm{\tau}) }
\bra{x, \sigma; \bm{\sigma}; f}D^{(\vepsilon_1)} \cdots D^{(\vepsilon_n)} \ket{y, \tau; \bm{\tau}; e^{-\beta \omega N_{\rm ph}}g}.\label{DLower}
\ee
From (iv) of  Lemma \ref{IroiroPP}, it follows that $e^{-\beta \omega N_{\rm ph}} g>0$, so if {\bf \hyperlink{R3}{(R. 3)}} holds, then 
we can find $n\in \BbbN$ and $\vepsilon_1,\dots, \vepsilon_n\in \{0, 1\}$ such that 
the right hand side of \eqref{DLower} is strictly positive. This completes the proof of Lemma \ref{R3Lemma}.
\end{proof}

\subsubsection*{\it Proof of Proposition \ref{Reduction}}
For any path $\bm{p}=(\{x_i, x_{i+1}\})_{i=1}^{n-1}$, we denote its length by $|\bm{p}|$: $|\bm{p}|=n-1$.
Since $D^{(0)}_{x, y} \unrhd 0$ and $D^{(1)}_x\unrhd 0$ w.r.t. $\fP_M$ due to Lemma \ref{IroiroPP}, it holds that 
\be
D^{(0)} \unrhd D^{(0)}_{x, y},\qquad D^{(1)}\unrhd D_x^{(1)}\quad \mbox{w.r.t. $\fP_M$}.
\ee
Accordingly, 
\begin{align}
&\bra{x, \sigma; \bm{\sigma}; f} (D^{(0)})^{|\bm{p}_1|} D^{(1)} (D^{(0)})^{|\bm{p}_2|} D^{(1)}\cdots (D^{(0)})^{|\bm{p}_n|} \ket{y, \tau; \bm{\tau}; g}\no
\ge& \bra{x, \sigma; \bm{\sigma}; f} (D^{(0)}_{\bm{p}_1} D_{x_1}^{(1)})( D^{(0)}_{\bm{p}_2} D_{x_2}^{(1)}) \cdots (D^{(0)}_{\bm{p}_m} D_{x_m}^{(1)}) D^{(0)}_{\bm{p}_{m+1}} \cdots D_{\bm{p}_n}^{(0)} \ket{y, \tau; \bm{\tau}; g}>0.
\end{align}
This  indicates that {\bf \hyperlink{R3}{(R. 3)}}  holds.
Therefore, by combining Lemmas \ref{R1Lemma}, \ref{R2Lemma} and \ref{R3Lemma}, we conclude the assertion in Proposition \ref{Reduction}.
\qed

\subsection{Proof of Theorem \ref{PI}}
What has been proved  is that to prove Theorem \ref{PI}, it suffices to show {\bf \hyperlink{R}{(R)}} in Porposition \ref{Reduction} holds.

Suppose that $x, y\in \vLa$ is given arbitrarily. It follows from the assumption \hyperlink{A2}{\bf (A. 2)} that there exists a path $\bm{p}=(\{x_i, x_{i+1}\})_{i=1}^{n-1}\subset E$ such that 
$y=x_1$ and $x=x_n$. Then we readily confirm that 
\be
D^{(0)}_{\bm{p}}\ket{x, \sigma; \bm{\sigma}; f}=t_{\bm{p}} \ket{y, \sigma; \bm{\sigma}; \varTheta_{\bm{p}} f},
\ee
where $t_{\bm{p}}=t_{x_1, x_2}t_{x_2, x_3}\cdots t_{x_{n-1}, x_n}>0$ and 
\be
\varTheta_{\bm{p}}=e^{{\rm i} \varPhi_{x_1, x_2}}e^{{\rm i} \varPhi_{x_2, x_3}} \cdots e^{{\rm i} \varPhi_{x_{n-1}, x_n}}.
\ee
From (iii) of Lemma \ref{IroiroPP}, if $f\ge 0$, then $\varTheta_{\bm{p}} f\ge 0$ holds; in addition,  since $\varTheta_{\bm{p}}$ is unitary, if $f\neq 0$, then $\varTheta_{\bm{p}} f \neq 0$ holds.
From the above, one can conclude the following:
\begin{description}
\item[Property 1.]The position $x$ of the conduction electron  in the state vector $\ket{x, \sigma; \bm{\sigma}; f}$ can be moved freely without breaking the positivity of this vector by the action of the operator $D_{\bm{p}}^{(0)}$.
\end{description}

Next, let us examine the action of the operator $D_x^{(1)}$. We readily confirm that 
\be
D_y^{(1)} \ket{y, \tau; \bm{\tau}; g}=J \ket{y, -\tau; T_x\bm{\tau}; g},
\ee
where $T_x$ denotes the spin-flip of the $f$-electron at site $x$;  see the proof of Lemma \ref{IroiroPP} for details.
From this,  we can see the following:

\begin{description}
\item[Property 2.]The spin configuration at site $x$ of the state vector $\ket{x, \sigma; \bm{\sigma}; f}$ can be reversed without breaking the positivity of this vector by the action of the operator $D^{(1)}_x$.
\end{description}

Using these two properties, for any $x, y\in \vLa $ and $ (\sigma, \bm{\sigma}), (\tau, \bm{\tau}) \in \bm{\mathcal{S}}_{\vLa, M}$, we will  construct an operator $\mathcal{C}((x, \sigma; \bm{\sigma}) \to (y, \tau; \bm{\tau}))$ from a  product of several  $D_{\bm{p}}^{(0)}$\rq{}s and $D_x^{(1)}$\rq{}s that satisfies the following properties: 
\be
\mathcal{C}((x, \sigma; \bm{\sigma}) \to (y, \tau; \bm{\tau}))\ket{x, \sigma; \bm{\sigma}; f}
=C \ket{y,  \tau; \bm{\tau}; \varTheta_{\bm{p}_1}\cdots \varTheta_{\bm{p}_n}f}, \label{ConnOp}
\ee
where $C$ is a constant expressed as the product of several $J$\rq{}s and $t_{x, y}$\rq{}s, and is particularly strictly positive;  $\bm{p}_1, \dots, \bm{p}_n$
are  paths used to construct the operator $\mathcal{C}((x, \sigma; \bm{\sigma}) \to (y, \tau; \bm{\tau}))$.
Because 
$g$ is stirictly positive and $\varTheta_{\bm{p}_1}\cdots \varTheta_{\bm{p}_n}f\ge 0$, if 
we can construct the  operator satisfying 
\eqref{ConnOp}, then we obtain
\begin{align}
\bra{y, \tau, \bm{\tau}; g}\mathcal{C}((x, \sigma; \bm{\sigma}) \to (y, \tau; \bm{\tau}))\ket{x, \sigma; \bm{\sigma}; f}
=C \la g|  \varTheta_{\bm{p}_1}\cdots \varTheta_{\bm{p}_n}f\ra>0.
\end{align}

From the discussion above, if we can construct an operator $\mathcal{C}((x, \sigma; \bm{\sigma}) \to (y, \tau; \bm{\tau}))$ of the form appearing in \eqref{Dprod}, the condition {\bf \hyperlink{R}{(R)}} is satisfied, thus completing the proof of Theorem \ref{PI}.
In the following, we will briefly describe its construction.   To this end, define the subset 
$B_{\bm{ \sigma}, \bm{ \tau}}$ of $\vLa$ as 
\be
B_{{\bm \sigma}, {\bm \tau}}=\{x\in \vLa : \sigma_x\neq \tau_x\}.
\ee
First, choose $z_1\in B_{{\bm \sigma}, {\bm \tau}}$ arbitrarily.
Let $\bm{p}_1$ be a path connecting  $x$ and $z_1$.  Then we have
\be
D^{(0)}_{\bm{p}_1}\ket{x, \sigma; \bm{\sigma}; f}=t_{\bm{p}_1}\ket{z_1, \sigma; \bm{\sigma}; \varTheta_{\bm{p}_1}f}.
\ee
When $D^{(1)}_{z_1}$ is further applied to the resulting state,  the following is obtained: 
\be
D^{(1)}_{z_1}D^{(0)}_{\bm{p}_1}\ket{x, \sigma; \bm{\sigma}; f}= Jt_{\bm{p}_1}\ket{z_1, -\sigma; T_{z_1}\bm{\sigma}; \varTheta_{\bm{p}_1}f}.
\ee
Note that the spin configuration of $\bm{\sigma}_1:=T_{z_1}\bm{\sigma}$  is equal to that of $\bm \tau$ at site $z_1$:
\be
 B_{{\bm \sigma}_1, {\bm \tau}}\subset B_{{\bm \sigma}, {\bm \tau}},\quad B_{{\bm \sigma}, {\bm \tau}}\setminus B_{{\bm \sigma}_1,  {\bm \tau}}=\{z_1\}.
\ee
Next, choose $z_2\in B_{{\bm \sigma}_1, {\bm \tau}}$ arbitrarily, and
let $\bm{p}_2$ be a  path connecting $z_1$ and $z_2$.
If we let the operator $D_{z_2}^{(1)} D^{(0)}_{\bm{p}_2}$ act on the state $\ket{z_1, -\sigma; \bm{\sigma}_1; \varTheta_{\bm{p}_1}f}$, we obtain 
\be
D_{z_2}^{(1)} D^{(0)}_{\bm{p}_2}\ket{z_1, -\sigma; \bm{\sigma}_1; \varTheta_{\bm{p}_1}f}=Jt_{\bm{p}_2}\ket{z_2, +\sigma;T_{z_2} \bm{\sigma}_1; \varTheta_{\bm{p}_2}\varTheta_{\bm{p}_1}f}.
\ee
Setting
$\bm{\sigma}_2=T_{z_2} \bm{\sigma}_1$, we have
\be
 B_{{\bm \sigma}_2, {\bm \tau}}\subset B_{{\bm \sigma}_1, {\bm \tau}},\quad B_{{\bm \sigma}_1, {\bm \tau}}\setminus B_{{\bm \sigma}_2,  {\bm \tau}}=\{z_2\}.
\ee
By repeating this procedure  until there is no element of $B_{\bm{\sigma}, \bm{\tau}}$ left, we can construct an operator $\mathcal{C}((x, \sigma; \bm{\sigma}) \to (y, \tau; \bm{\tau}))$ of the form appearing in \eqref{Dprod}.
This completes the proof of Theorem \ref{PI}. \qed

\subsection{Proof of Theorem \ref{Main2}}

Let $\tilde{\psi}_M$ be the ground state of $\tH_M$.
Due to Theorems \ref{pff} and \ref{PI}, $\tilde{\psi}_M>0$ w.r.t. $\fP_M$ holds.
In addition, by   the arguments similar to those in the proof of Lemma \ref{IroiroPP}, we have
\be
s_x^{(+)}s_y^{(-)} \unrhd 0,\ s_x^{(+)} S_y^{(-)} \unrhd 0,\ S_x^{(+)}S_y^{(-)} \unrhd 0\ \mbox{w.r.t. $\fP_M$}\ (x, y\in \vLa).
\ee
Combining these facts, we find 
\be
\bra{\tilde{\psi}_M}
s_x^{(+)}s_y^{(-)} \tilde{\psi}_M\ra > 0,\ \bra{\tilde{\psi}_M}s_x^{(+)} S_y^{(-)} \tilde{\psi}_M\ra > 0,\ \bra{\tilde{\psi}_M}S_x^{(+)}S_y^{(-)} \tilde{\psi}_M\ra > 0. \label{TiExp}
\ee
On the other hand, we readily confirm that 
\be
F^{-1}s_x^{(+)}s_y^{(-)} F=s_x^{(+)}s_y^{(-)},\ F^{-1} s_x^{(+)} S_y^{(-)}F=-s_x^{(+)} S_y^{(-)},\  F^{-1}S_x^{(+)}S_y^{(-)}F=S_x^{(+)}S_y^{(-)}.
\ee
Putting \eqref{TiExp} together with this fact, we obtain the assertion of Theorem \ref{Main2}.
\qed

\section{Proof of Theorem \ref{Main3}}\label{Sec4}
We prove Theorem \ref{Main3} by extending the ideas in \cite{Miyao2017}.
By straightforward computation, we have
\begin{align}
s_x^{(3) }S_x^{(3)} |{\bs \sigma}_x\ra_0&=0,\\
s_x^{(+) }S_x^{(-)} |{\bs \sigma}_x\ra_0&=-\frac{1}{\sqrt{2}} c_{x, \up}^* f_{x, \down}^* |{\bs \sigma}_x\ra,\\ 
s_x^{(-) }S_x^{(+)} |{\bs \sigma}_x\ra_0&=\frac{1}{\sqrt{2}} c_{x, \down}^* f_{x, \up}^* |{\bs \sigma}_x\ra.
\end{align}
If $x\neq y$, then 
\be
s_x^{(3) }S_x^{(3)} |{\bs \sigma}_y\ra_0=
s_x^{(+) }S_x^{(-)} |{\bs \sigma}_y\ra_0=
s_x^{(-) }S_x^{(+)} |{\bs \sigma}_y\ra_0=0,
\ee
so we get
\be
J\sum_{x\in \vLa} {\bs s}_x\cdot {\bs S}_x|{\bs \sigma}_y\ra_0=-\frac{J}{2} |{\bs \sigma}_y\ra_0.\label{ActEx1}
\ee

Next, for each 
${\bs \sigma}=(\sigma_y)_{y\in \vLa}\in \mathcal{S}_{\vLa}$, we set 
\begin{align}
| {\bs \sigma}_x\ra_{1, 1}&=
c_{x, \up}^* f_{x, \up}^*
 |{\bs \sigma}_x\ra,\\
  |{\bs \sigma}_x\ra_{1, 2}&=
c_{x, \down}^* f_{x, \down}^*
 |{\bs \sigma}_x\ra,\\
|{\bs \sigma}_x\ra_{1, 3}&=
\frac{1}{\sqrt{2}} (c_{x, \up}^* f_{x, \down}^*+c_{x, \down}^*f_{x, \up}^*)
 |{\bs \sigma}_x\ra.
\end{align}
These vectors represent situations where the conduction electron and $f$-electron  form a triplet at site $x$. The subscript $1$ denotes this fact. We readily confirm that 
\be
\{ |{\bs \sigma}_x\ra_0,\ | {\bs \sigma}_x\ra_{1, j} : j=1, 2, 3, {\bs \sigma}\in \mathcal{S}_{\vLa}, x\in \vLa\}
\ee
 is a CONS of $\h$.
By performing the similar calculations as before, we have
\begin{align}
J\sum_{x\in \vLa} {\bs s}_x\cdot {\bs S}_x| {\bs \sigma}_y\ra_{1, j}=\frac{J}{4} | {\bs \sigma}_y\ra_{1, j},\ \ j=1, 2, 3. \label{ActEx2}
\end{align}
Furthermore, 
\be
\mathfrak{X}^{\perp}=\mathrm{Lin}\{|{\bs \sigma}_x\ra_{1, j}\otimes \vphi : j=1,2,3, {\bs \sigma}\in \mathcal{S}_{\vLa}, x\in \vLa, \vphi\in \h_{\rm ph}\}
\ee
holds, where $\mathfrak{X}$ is  defined by \eqref{DefX}.

To study the limit of $J\to \infty$, one considers the renormalized Hamiltonian:
\be
H_{{\rm ren}, J}=H+\frac{J}{2}.
\ee
To simplify the discussion, we introduce the following operators:
\be
H_{\infty}=PH_{{\rm ren}, J}P, \ \ H_1=P^{\perp}H_{{\rm ren}, J} P^{\perp},\ \ H_{01}=PH_{{\rm ren}, J}P^{\perp}+
P^{\perp}H_{{\rm ren}, J}P,
\ee
where $P$ is the orthogonal projection from $\h$ to $\mathfrak{X}$ and $P^{\perp}=\mathbbm{1}-P$. 
Because the spin exchange interaction term commutes with  the interaction term between the  spins and the external magnetic field, we see
\begin{align}
H_{01}=PTP^{\perp}+P^{\perp}TP, 
\end{align}
where $T$ stands for the   hopping term:
\be
T=\sum_{x, y\in \vLa}\sum_{\sigma=\up, \down}(-t_{x, y}) c_{x, \sigma}^*c_{y, \sigma}.
\ee

Define
\be
R=T_h+H_{\rm ph},
\ee
where 
\be
T_h=T-2hS_{\rm tot}^{(3)},\quad
H_{\rm ph} =\omega N_{\rm ph}+\sum_{x, y \in \vLa} g_{x, y}n_x^c(b_y + b_y^*).
\ee
Then, from the above arguments, we know the following:
\begin{Lemm}\label{EnyX}
For a given self-adjoint operator $A$, bounded from below, set
$E(A)=\inf \mathrm{spec}(A)$. We have the following:
\begin{itemize}
\item[\rm (i)] $\displaystyle
E(H_{\infty} \restriction \mathfrak{X})=E(P R  P)$.
\item[\rm (ii)] $\displaystyle 
E(H_1 \restriction \mathfrak{X}^{\perp})=E(P^{\perp} R P^{\perp})+\frac{3J}{4}
$.
\end{itemize}
\end{Lemm}
\begin{proof}
Denote by 
$H_{\rm ex}$ the spin-exchange interaction term: $H_{\rm ex}=-J\sum_{x\in \vLa}{\bm{s}}_x\cdot \bm{S}_x$.
Due to \eqref{ActEx1} and \eqref{ActEx2}, $H_{\rm ex}$ commutes with $P$, and 
\be
H_{\rm ex}P=-\frac{J}{2}P,\ \ H_{\rm ex}P^{\perp} =\frac{J}{4}P^{\perp}. 
\ee
Hence, we obtain
\be
H_{\infty}=PRP,\qquad
H_1=P^{\perp} R  P^{\perp}+\frac{3J}{4}P^{\perp}.
\ee
Consequently, (i) and (ii)  of  Lemma \ref{EnyX} immediately follow.
\end{proof}

From (ii) of Lemma \ref{EnyX}, the following corollary follows:
\begin{Coro}\label{JBound}
If $J$ is sufficiently large, then $H_1^{\perp} P^{\perp}$ is a bounded operator.
\end{Coro}
\begin{proof}
Using the commutation relations \eqref{CCRs}, we have
$$
\|b_x \vphi \|\le \|N_{\rm ph}^{1/2} \vphi\|,\quad \|b_x^*\vphi\|\le \|(N_{\rm ph}+\mathbbm{1})^{1/2} \vphi\|
\quad (\vphi\in \D(N_{\rm ph}^{1/2})),
$$
which implies
\be
\bigg\|\sum_{x, y\in \vLa} g_{x, y} n_x^c(b_y+b_y^*)\vphi\bigg\|\le C_g \|(N_{\rm ph}+\mathbbm{1})^{1/2} \vphi\|
\quad (\vphi\in \D(N_{\rm ph}^{1/2})),
\ee
where $C_g=\sum_{x, y}|g_{x, y}|$. Combining this  with the elementary inequality
$ab \le \vepsilon a^2+b^2/4\vepsilon\ (\vepsilon>0)$, we obtain
\be
\la \vphi|H_{\rm ph}\vphi\ra\ge (\omega-\vepsilon C_g)\la \vphi|N_{\rm ph} \vphi\ra-\frac{C_g}{4\vepsilon} \|\vphi\|^2\quad(\vphi\in\D(N_{\rm ph}^{1/2})).
\ee
Therefore, choosing $\vepsilon>0$ such that $\omega-\vepsilon C_g\ge 0$, we get
\be
\bra{\vphi} R\vphi\ra\ge -C\|\vphi\|\quad (\vphi\in \D(N_{\rm ph}^{1/2})), 
\ee
 where $C=\sum_{x, y\in \vLa} |t_{x, y}|+C_g/4\vepsilon$. 
This together with (ii) of Lemma \ref{EnyX} yields 
$E(H_1\restriction \mathfrak{X}^{\perp}) \ge -C+3J/4$. Note that $C$ does not depend on $J$.
Because $E(H_1\restriction \mathfrak{X}^{\perp})$ is strictly positive whenever  $J>4C/3$, we obtain the desired assertion in the corollary.
\end{proof}

For simplicity of notation,  set 
$
\mathcal{E}(H_1)=E(H_1 \restriction \mathfrak{X}^{\perp}).
$
According to (ii) of Lemma \ref{EnyX}, it follows that 
\be
\lim_{J\to \infty}\mathcal{E}(H_1)=\infty.\label{JtoInf}
\ee
This fact is used repeatedly in the following discussion.
\begin{Lemm}\label{U3}
Let  $z\in \BbbC\backslash \BbbR$.  If $J$ is sufficiently large enough, we have 
\begin{align}
\|(H_1-z)^{-1} P^{\perp}\| \le  \{\mathcal{E}(H_1)-|z|\}^{-1}.
\end{align} 
\end{Lemm} 
\begin{proof} Since  $H_1^{-1}P^{\perp}$ is bounded due to Corollary \ref{JBound}, we have
\begin{align}
(H_1-z)^{-1}P^{\perp}=\sum_{n=0}^{\infty} (H_1^{-1}P^{\perp} z)^n H_1^{-1}P^{\perp},
\end{align} 
which implies that 
\be
\|(H_1-z)^{-1 }P^{\perp}\| \le \sum_{n=0}^{\infty} \mathcal{E}(H_1)^{-n-1} |z|^n=\{\mathcal{E}(H_1)-|z|\}^{-1}.
\ee
\end{proof}

\begin{Lemm}\label{U4}
Let $z\in \BbbC\setminus \BbbR$. If $|\mathrm{Im} z|$ is 
 large enough, then we have
\begin{align}
 \lim_{J \to \infty}
\Big\|
(H-z)^{-1}-(H_{\infty}+H_1-z)^{-1}
\Big\|=0.
\end{align} 
\end{Lemm} 
\begin{proof} First, remark that 
\begin{align}
\Big\{
(H-z)^{-1}-(H_{\infty}+H_1-z)^{-1}
\Big\}P
=(H-z)^{-1}P^{\perp} (-H_{01}) (H_{\infty}-z)^{-1}P. \label{Exp1}
\end{align} 
The norm of $(H-z)^{-1} P^{\perp}$ is estimated as follows:
Since 
\begin{align}
(H-z)^{-1} P^{\perp}=\sum_{n=0}^{\infty}(-1)^n
 \Big\{(H_{\infty}+H_1-z)^{-1} H_{01}\Big\}^n
(H_1-z)^{-1}P^{\perp},
\end{align} 
we have, by using Lemma \ref{U3},
\begin{align}
\|(H-z)^{-1} P^{\perp}\|& \le 
\Bigg(
\sum_{n=0}^{\infty} |\mathrm{Im} z|^{-n} \|H_{01}\|^n
\Bigg) \{\mathcal{E}(H_1)-|z|\}^{-1}\no
&:=C_z\{\mathcal{E}(H_1)-|z|\}^{-1},
\end{align} 
where
 $z$ is chosen such that $|\mathrm{Im} z|^{-1}\|H_{01}\|<1$. Note that since $\|H_{01}\|$ does not depend on $J$, so too $C_z$ does not depend on $J$.
Thus, by using  \eqref{JtoInf} and  (\ref{Exp1}), we find 
\begin{align}
\Bigg\| 
\Big\{
(H-z)^{-1}-(H_{\infty}+H_1-z)^{-1}
\Big\}P
\Bigg\| \le C_z\|H_{01}\| |\mathrm{Im} z|^{-1}\{\mathcal{E}(H_1)-|z|\}^{-1} \to 0
\end{align} 
as $J\to \infty$.

Next, together \eqref{JtoInf} and  Lemma \ref{U3} with the following identity:
\begin{align}
\Big\{
(H-z)^{-1}-(H_{\infty}+H_1-z)^{-1}
\Big\}P^{\perp}=(H-z)^{-1} P(-H_{01}) (H_1-z)^{-1} P^{\perp},
\end{align} 
we obtain
\begin{align}
\Bigg\|
\Big\{
(H-z)^{-1}-(H_{\infty}+H_1-z)^{-1}
\Big\}P^{\perp}
\Bigg\| \le |\mathrm{Im} z|^{-1}\|H_{01}\| \{\mathcal{E}(H_1)-|z|\}^{-1} \to 0
\end{align}  
as $J\to \infty$. \end{proof}

\begin{Lemm}\label{U5}
Let $z\in \BbbC\setminus \BbbR$. If $|\mathrm{Im} z|$ is 
 large enough, then we have
\begin{align}
\lim_{J\to \infty}\Big\|
(H_{\infty}+H_1-z)^{-1}-(H_{\infty}-z)^{-1}P
\Big\|=0.
\end{align} 
\end{Lemm} 
\begin{proof}
Remark  that 
\begin{align}
(H_{\infty}+H_1-z)^{-1}-(H_{\infty}-z)^{-1}P=
(H_1-z)^{-1}P^{\perp}.
\end{align} 
Hence, we obtain, by applying \eqref{JtoInf} and Lemma \ref{U3},
\begin{align}
\Big\|
(H_{\infty}+H_1-z)^{-1}-(H_{\infty}-z)^{-1}P
\Big\| \le \{\mathcal{E}(H_1)-|z|\}^{-1}\to 0
\end{align} 
as $J\to \infty$.  
\end{proof}

\begin{flushleft}{\it
Completion of  the proof of Theorem \ref{Main3}
}
\end{flushleft} 

By applying Lemmas \ref{U4}, \ref{U5} and \cite[Theorem VIII. 19]{Reed1981}, we obtain the desired result in the
theorem. $\Box$

\section{Proof of Theorems \ref{Main4} and \ref{Main5}}\label{Sec5}
\subsection{Equivalence with the Nagaoka--Thouless system}
The outline of the proof of Theorems \ref{Main4} and \ref{Main5} is as follows:  (i) we show that the renormalized Hamiltonian $H_{\rm ren, \infty}$ is unitarily equivalent to the Nagaoka--Thouless Hamiltonian $H_{\rm NT}$ in the electron-phonon interaction system, and (ii) we apply existing results concerning $H_{\rm NT}$.
In this subsection, we show the first step of the proof, the unitary equivalence of $H_{\rm ren, \infty}$ and $H_{\rm NT}$.

Consider a system with $|\vLa|-1$ electrons on the lattice $\vLa$.
The Hilbert space describing  this system is given by
$\bigwedge^{|\vLa|-1}\big( \ell^2(\vLa)\oplus  \ell^2(\vLa) \big)$.
The electron creation and annihilation operators of this system are denoted by $d_x^*$ and $ d_x$, respectively. These satisfy the standard anticommutation relations:
\be
\{d_{x, \sigma}, d_{y, \tau}\}=0, \quad \{d_{x, \sigma}, d_{y, \tau}^*\}=\delta_{\sigma, \tau} \delta_{x, y}.
\ee
The Nagaoka--Thouless (NT) system is a many-electron system in which there is precisely a single hole and all other sites are occupied by a single electron.
In order to mathematically describe the NT system, we introduce 
the  Gutzwiller projection by 
\begin{align}
Q=\prod_{x\in \vLa} (\mathbbm{1}-n^d_{x,\uparrow}n^d_{x,\downarrow}).
\end{align} 
$Q$ is the orthogonal projection onto
the subspace with no doubly occupied sites.
The Hilbert space:
\be
Q \bigwedge^{|\vLa|-1}  \big(\ell^2(\vLa)\oplus  \ell^2(\vLa)\big)
\ee
describes the states of the ordinary NT system.
The operator 
\be
K_h=\sum_{x, y\in \vLa}  \sum_{\sigma=\up, \down}  b_{x, y} Q d_{x, \sigma}^*d_{y, \sigma} Q-2h S^{(3)}_{{\rm tot}, d}Q
\ee
acting on this Hilbert space is referred to as the {\it Nagaoka--Thouless Hamiltonian},\footnote{Various studies have been conducted on the NT Hamiltonian. 
For an extension in a different direction from this paper, see \cite{PhysRevA.87.013617}.} where
\be
S_{{\rm tot}, d}^{(3)}=\frac{1}{2}\sum_{x\in \vLa}( d_{x, \up}^*d_{x, \up}-d_{x, \down}^*d_{x, \down}).
\ee
We are interested in the interaction of phonons with the NT system. As a Hamiltonian describing such a system, we consider 
\be
H_{\rm NT}=K_h+\sum_{x, y\in \vLa} g_{x, y}n_x^d(b_y+b_y^*)Q+\omega N_{\rm ph},
\ee
where $n_x^d=\sum_{\sigma=\up, \down} d_{x, \sigma}^*d_{x, \sigma}$.
$H_{\rm NT}$ is a self-adjoint operator, bounded from below,  acting  in the Hilbert space:
\be
\h_{\rm NT}=Q \bigwedge^{|\vLa|-1}  \big(\ell^2(\vLa)\oplus  \ell^2(\vLa)\big) \otimes \h_{\rm ph}.
\ee

\begin{Rem} \upshape
This Hamiltonian can be  derived from the Holstein--Hubbard Hamiltonian by taking the limit of the strength of the Coulomb interaction $U$ to infinity.
\end{Rem}

For each $x\in \vLa$ and $\bs \sigma=(\sigma_x)_{x\in \vLa}\in  \mathcal{S}_{\vLa}$,  set
\be
|{\bs \sigma}_x\ra_{\rm NT}=d_{x, \sigma_x} \prod^{\sharp}_{x\in \vLa} d_{x, \sigma_x}^*|\varnothing\ra.
\ee
We define the unitary operator $W: \mathfrak{X} \to \h_{\rm NT}$ by 
\be
W |{\bs \sigma_x}\ra_0\otimes \vphi=|\bs \sigma_x\ra_{\rm NT}\otimes \vphi\quad (\bm{\sigma}\in \mathcal{S}_{\vLa},\ \vphi\in \h_{\rm ph}),
\ee
where $ |{\bs \sigma_x}\ra_0$ is given by \eqref{SingV}.
\begin{Prop}\label{NTEquiv}
Choosing $b_{x, y}=\frac{1}{2}t_{x, y}$, one has
\be
W H_{\rm ren, \infty} W^*=H_{\rm NT}.
\ee

\end{Prop}
\begin{proof}
For each 
$x\in \vLa$ and $\bs \sigma\in \mathcal{S}_{\vLa}$, we readily confirm 
\be
c_{\mu,\up}|\bs \sigma_x\ra_0=
\begin{cases}
\frac{1}{\sqrt{2}} f_{x, \up}^*|\bs \sigma_x\ra & \mbox{ if $x=\mu$}\\
0 & \mbox{otherwise}
\end{cases},
\ \ 
c_{\mu,\down}|\bs \sigma_x\ra_0=
\begin{cases}
-\frac{1}{\sqrt{2}} f_{x, \down}^*|\bs \sigma_x\ra & \mbox{ if $x=\mu$}\\
0 & \mbox{otherwise.}
\end{cases}
\ee
Using these, we find
\begin{align}
{}_0\la \bs \sigma_x|T_h|{\bs \tau_y}\ra_0
=-\frac{1}{2} \sum_{\sigma=\up, \down}
t_{x, y} \delta_{{\bs \sigma_x}\cup \{\sigma\}_x,\ \bs \tau_y\cup \{\sigma\}_y}-h\sum_{z\in \vLa\setminus \{x\}} \sigma_z\delta_{x, y} \delta_{{\bs \sigma}_x, {\bs \tau}_y},
\end{align}
where,  for each $\bs \sigma=(\sigma_y)_y\in \mathcal{S}_{\vLa}$, we  define $\bs \sigma_x\cup \{\sigma\}_x=(\xi_y)_{y\in \vLa}\in \mathcal{S}_{\vLa}$ by
\be
\xi_y=\begin{cases}
\sigma_y & \mbox{if $y\neq x$}\\
\sigma & \mbox{if $y=x$}.
\end{cases}
\ee
On the other hand, using the following fact:
\be
d_{\mu, \sigma}^*|\bs\sigma_x\ra_{\rm NT}=
\begin{cases}
|\bs \sigma_x\cup\{\sigma\}_x\ra & \mbox{if $\mu=x$}\\
0& \mbox{otherwise,}
\end{cases}
\ee
we obtain
\begin{align}
{}_{\rm NT}\la \bs \sigma_x|K_h|\bs \tau_y\ra_{\rm NT}
=-\sum_{\sigma=\up, \down} b_{x, y} \delta_{{\bs \sigma_x}\cup \{\sigma\}_x,\ \bs \tau_y\cup \{\sigma\}_y}
-h\sum_{z\in \vLa\setminus \{x\}} \sigma_z\delta_{x, y} \delta_{{\bs \sigma}_x, {\bs \tau}_y}.
\end{align}
Therefore, by choosing $b_{x, y}=\frac{1}{2}t_{x, y}$, we get
$
{}_0\la  \bs \sigma_x|T_h| {\bs \tau_y}\ra_0={}_{\rm NT}\la \bs \sigma_x|K_h|\bs \tau_y\ra_{\rm NT}
$.
Similarly, we can show that 
\be
{}_0\la  \bs \sigma_x;  f |\sum_{x, y\in \vLa} n_x^c(b_y+b_y^*) | {\bs \tau_y} ; g\ra_0
={}_{\rm NT}\la  \bs \sigma_x;  f |\sum_{x, y\in \vLa} n_x^d(b_y+b_y^*) | {\bs \tau_y} ; g\ra_{\rm NT},
\ee
where we set
\be
\ket{\bm{\sigma}_x; f}_0=\ket{\bm{\sigma}_x}_0\otimes f,\quad
\ket{\bm{\sigma}_x; f}_{\rm NT}=\ket{\bm{\sigma}_x}_{\rm NT}\otimes f\quad (f\in \h_{\rm ph}).
\ee
To summarize the arguments so far, we obtain
\be
{}_0\la  \bs \sigma_x; f|H_{{\rm ren}, \infty}| {\bs \tau_y}; g\ra_0={}_{\rm NT}\la \bs \sigma_x; f|H_{\rm NT}|\bs \tau_y; g\ra_{\rm NT}\label{IdnNt}
\ee
for any $x, y\in \vLa, \bm{\sigma}, \bm{\tau}\in \mathcal{S}_{\vLa}$ and $f, g\in \D(N_{\rm ph})$.
From the definition of $W$ and \eqref{IdnNt}, we obtain the desired assertion.
\end{proof}

\subsection{Proof of Theorems \ref{Main4} and \ref{Main5}}

In \cite{Miyao2017}, the following theorem is proved:
\begin{Thm}\label{NT1} 
Assume \hyperlink{A1}{\bf (A. 1)}.
Suppose that the hopping matrix $(b_{x, y})$ of $H_{\rm NT}$ satisfies the similar condition as \hyperlink{A3}{\bf (A. 3)}.\footnote{It is proved in \cite{PhysRevB.98.180101} that the condition \hyperlink{A3}{\bf (A. 3)} is equivalent to the connectivity condition in \cite{Miyao2017}. For a clear explanation concerning the connectivity condition, see \cite{Tasaki2020}.}
 Then the ground state of $H_{\rm ren,\infty}$ is unique apart from the trivial $|\vLa|$-fold degeneracy, and has total spin  $S=(|\vLa|-1)/2$.
 \end{Thm}

On the other hand, the following theorem is proved in \cite{Miyao2020-2}:
\begin{Thm}\label{NT2}
Let $d$ be a natural number greater than or equal to $2$ and consider the case where $\vLa=[-L, L)^d\cap \BbbZ^d$.
In addition, assume that the hopping matrix $(b_{x, y})$ of $H_{\rm NT}$ expresses the nearest neighbor hopping.
 Let $M_{\rm NT}(\beta, h)$ be the magnetization defined by replacing $H$ with $H_{\rm NT}$ in  \eqref{DefMg}.
 Then one obtains 
\be
M_{\rm NT}(\beta, h)\ge (|\vLa|-1) \tanh (\beta h).
\ee

\end{Thm}

Using  these theorems with Proposition \ref{NTEquiv}, we can prove Theorems \ref{Main4} and \ref{Main5}.\qed

\appendix

\section{Brief overview of the definition of stability classes}\label{AppA}

The theory of stability classes in many-electron systems is originated in  \cite{Miyao2019}. A more mathematically refined construction of the theory is presented in \cite{Miyao2022}. The essence of the idea is as follows.

Let $\mathfrak{H}_1$ be a Hilbert space and $\mathfrak{H}_0$ be its closed subspace.
Suppose that the total spin operators $\bs S_{\rm tot}$ act on the two Hilbert spaces, and furthermore that the two  Hilbert spaces $\h_0$ and $\h_1$ are subspaces of the $M=0$-subspace: $\ker(S^{(3)}_{\rm tot})$.

Let $P_{1, 0}$ be the orthogonal  projection  from $\mathfrak{H}_1$ to $\mathfrak{H}_0$.
Moreover, let $\mathfrak{P}_0$ and $\mathfrak{P}_1$ be Hilbert cones in $\mathfrak{H}_0$ and $\mathfrak{H}_1$, respectively.
Now, suppose that 
\be
P_{1, 0}\mathfrak{P}_1=\mathfrak{P}_0 \label{PtoP}
\ee
holds. Assume that $\psi_0\in \mathfrak{P}_0$ is strictly positive and has total spin $S_0$: ${\bs S}^2_{\rm tot}\psi_0=S_0(S_0+1) \psi_0$.
Next, assume that a vector $\psi_1$ in $\h_1$ is strictly positive with respect to $\mathfrak{P}_1$ and has total spin $S_1$.
With these settings, we claim that 
\be
S_0=S_1\label{SeqS}
\ee
holds. Indeed, because $P_{1, 0}\psi_1\neq 0$ and $P_{1, 0}\psi_1\ge 0$ w.r.t. $\mathfrak{P}_0$ hold due to \eqref{PtoP}, 
we get the  positive overlap property: $\la \psi_0|P_{1, 0} \psi_1\ra>0$, which implies that 
\be
S_0(S_0+1) \la \psi_0|P_{1, 0} \psi_1\ra=\la {\bs S}^2_{\rm tot}\psi_0|P_{1, 0}\psi_1\ra
=\la\psi_0|P_{1, 0} {\bs S}^2_{\rm tot}\psi_1\ra=S_1(S_1+1) \la \psi_0|P_{1, 0} \psi_1\ra.
\ee 
Hence, we conclude \eqref{SeqS}.

Let $\mathscr{C}(\mathfrak{P}_0)$ be the set of pairs $(\h, \mathfrak{P})$ of a Hilbert space and a Hilbert cone satisfying the following conditions:
\begin{itemize}
\item $\mathfrak{P}$ is a Hilbert cone in $\h$;
\item $\h_0$ is a subspace of $\h$;
\item a similar relation to \eqref{PtoP}: 
\be
P \mathfrak{P}=\mathfrak{P}_0
\ee
is fulfilled, where $P$ is the orthogonal projection from $\h$ to $\h_0$.
\end{itemize}
If we choose arbitrarily a pair $(\h, \mathfrak{P})\in \mathscr{C}(\mathfrak{P}_0)$, and assume that a vector $\psi\in \h$ is strictly positive with respect to  $\mathfrak{P}$ and has total spin $S$, then we conclude $S=S_0$ using the positive overalp property  as in the previous discussion.

Now, suppose we are given a Hamiltonian $H$ acting on $\h$. Assume then that this Hamiltonian satisfies the following conditions:
\begin{itemize}
\item[(i)] $H$ commutes with the total spin operators $S_{\rm tot}^{(j)}\ (j=1, 2, 3)$.
\item[(ii)] The heat semigroup $\{e^{-\beta H}\}$ is ergodic w.r.t. $\mathfrak{P}$.
\end{itemize}
Theorem \ref{pff}  shows that the ground state $\psi_g$ of $H$ is strictly positive with respect to $\mathfrak{P}$, and consequently, we know that $\psi_g$ has   total spin $S_0$ due to the positive overlap property.

Now, let $\mathscr{A}_{\mathfrak{P}_0}(\h, \mathfrak{P})$ be the set of  all Hamiltonians satisfying the two conditions (i) and (ii).
Then the ground state of any Hamiltonian belonging to $\mathscr{A}_{\mathfrak{P}_0}(\h, \mathfrak{P})$ has total spin $S_0$.
The {\it  $\mathfrak{P}_0$-stability class}  is defined by 
\be
\mathscr{A}_{\mathfrak{P}_0}=\bigcup_{(\h, \mathfrak{P}) \in \mathscr{C}(\mathfrak{P}_0)} \mathscr{A}_{\mathfrak{P}_0}(\h, \mathfrak{P}).
\ee
From the construction, we see  that the ground states of all Hamiltonians belonging to $\mathscr{A}_{\mathfrak{P}_0}$ have total spin $S_0$.

This theory describes the stability of magnetic properties of ground states in many-electron systems.
To a given Hamiltonian $H_1$, let $H_2$ be a Hamiltonian obtained by adding complicated interaction terms with the environment typified by phonons. In general, the analysis of $H_2$ becomes more difficult, but if we find that $H_1$ and $H_2$ both belong to the same stability class $\mathscr{A}_{\mathfrak{P}_0}$, we conclude that the ground states of these Hamiltonians have the same total spin $S_0$.

Define
\be
\fP_{{\rm NT}}={\rm coni}\Bigg(\bigg\{\ket{\bm{\sigma}_x}_{\rm NT} : \bm{\sigma}\in \mathcal{S}_{\vLa}, \sum_{y\in \vLa\setminus \{x\}} \sigma_y=0, x\in \vLa \bigg\}\Bigg).
\ee
Then $\fP_{{\rm NT}}$ is a Hilbert cone in the Hilbert space $Q \bigwedge^{|\vLa|-1} \big(\ell^2(\vLa)\oplus \ell^2(\vLa)\big) \cap \ker(S_{\rm tot}^{(3)})$ that emerged in defining the NT system.\footnote{In Section \ref{Sec5}, the total spin operator in the NT system was denoted by  $S^{(3)}_{{\rm tot}, d}$ to avoid confusion with the total spin operators in the previous sections. We do not distinguish between the two here.}
The $\mathfrak{P}_{\rm NT}$-stability class constructed from $\fP_{\rm NT}$ is called the {\it Nagaoka--Thouless stability class}.
Theorem \ref{NTGRP} and other fundamental properties of the NT stability class are discussed in detail in \cite{Miyao2019}.
Under these settings, the meaning of Theorem \ref{NTstatThm} should now be clear.
From this theorem, it immediately follows that the values of total spin in the ground state coincide in Theorem \ref{Main1} and Theorem \ref{Main4}.

In addition to the NT stability class discussed here, several other stability classes are known. Recently, it has become clear that the stability theory can describe the flat-band ferromagnetism \cite{Miyao-Tominaga2022}. On the other hand, some progress has been made in grounding this theory with the Tomita--Takesaki theory in operator algebras \cite{Miyao2022}.

\end{document}